\newif\iffullversion
\fullversiontrue

\newif\ifacmversion

\newif\ifusxversion

\newif\ifspversion

\iffullversion
\documentclass[conference,compsoc]{IEEEtran}
\fi
\ifacmversion
\documentclass[sigconf,review]{acmart}
\fi
\ifusxversion
\documentclass[letterpaper,twocolumn,10pt]{article}
\usepackage{usenix-2020-09}
\fi
\ifspversion
\documentclass[conference,compsoc]{IEEEtran}
\fi

\AtBeginDocument{%
  }

\ifacmversion
\setcopyright{acmlicensed}
\copyrightyear{2024}
\acmYear{2024}
\acmDOI{XXXXXXX.XXXXXXX}

\acmConference[ACM CCS '24]{
}{October 14--18, 2014}{Salt Lake City, U.S.A.}

\usepackage[T1]{fontenc}
\includeversion{attributed}
\fi

\usepackage[underline=false]{pgf-umlsd}

\usepackage{tikz}

\usepackage{comment}

\ifspversion
 \renewcommand{\paragraph}[1]{
 \noindent
     \textbf{#1} 
 }
\fi
\iffullversion
 \renewcommand{\paragraph}[1]{
 \noindent
     \textbf{#1} 
 }
\fi
\usepackage[primitives]{cryptocode}
\usepackage{pgfplots}
\pgfplotsset{compat=1.18} 
\usepgfplotslibrary{groupplots}

\usepackage{caption}
\usepackage{subcaption}
\usepackage{xurl}   
\usepackage{xcolor, colortbl}

\usepackage{algorithm}
\usepackage[noend]{algpseudocode}
\usepackage{amsthm}
\newtheorem{theorem}{Theorem}[section]
\newtheorem{lemma}[theorem]{Lemma}

\newtheorem{example}{Example}[section]
\newtheorem{definition}{Definition}[section]

\usepackage{versions}

\usepackage{float}

\usepackage{enumitem}
\usepackage[normalem]{ulem}
\usepackage{amsmath}
\usepackage{amsthm}
\usepackage{amsfonts}
\usepackage{comment}
\usepackage{algpseudocodex,algorithmicx}

\usepackage{booktabs,multirow}
\usepackage{stmaryrd}
\usepackage{url}
\usetikzlibrary{positioning,quotes,arrows,shapes,decorations.text}
\usepackage{circledsteps}
\usepackage{fontawesome5}

\usepackage{hyperref}
\hypersetup{
   colorlinks=true,        
   linkcolor=blue,         
    citecolor=green,        
    urlcolor=cyan,           
    hypertexnames=false 
}
\usepackage{cleveref}

\algblockdefx[procedure]{procedure}{EndProcedure}[2]{\textbf{procedure} \textsc{#1}(#2)}{\algpx@endIndent}
\algtext*{EndProcedure}
\let\oldReturn\Return
\renewcommand{\Return}{\State\oldReturn}

\ifspversion
    \ifCLASSOPTIONcompsoc
      \usepackage[nocompress]{cite}
    \else
      \usepackage{cite}
    \fi
\fi
\iffullversion
\usepackage[nocompress]{cite}
\fi
\usepackage[textsize=footnotesize,textwidth=1in]{todonotes}
\newcommand{\sv}[1]{{\color{purple} [SV: #1]}}
\newcommand{\mm}[1]{{\color{blue} [MM: #1]}}
\newcommand{\ram}[1]{{\color{cyan} [RAM: #1]}}

\ifacmversion
    \usepackage{titlesec}
    \titlespacing*{\paragraph}{0pt}{0.3ex plus 0.2ex minus .2ex}{1em}
\fi

\newcommand{\peerid}{peerID}
\newcommand{\peerids}{peerIDs}

\newcommand{\sample}{\xleftarrow{\smash{\raisebox{-1.75pt}{$\scriptscriptstyle\$$}}}}
\newcommand{\getsr}{\sample}

\newcommand{\protstyle}[1]{\textsc{#1}}   
\newcommand{\algostyle}[1]{\mathsf{#1}}   
\newcommand{\protostyle}[1]{\textsc{#1}}

\newcommand{\ipfskadem}{\protstyle{Kademlia}}
\newcommand{\privipfskadem}{\protstyle{priv-Kademlia}}
\newcommand{\sha}{\protstyle{SHA-256}}

\newcommand{\encrypt}{\algostyle{Enc}}
\newcommand{\decrypt}{\algostyle{Dec}}
\newcommand{\add}{\algostyle{Add}}
\newcommand{\scalmult}{\algostyle{ScalMult}}

\newcommand{\oblivexpand}{\algostyle{ObliviousExpand}}

\newcommand{\query}{\algostyle{Query}}
\newcommand{\response}{\algostyle{Response}}
\newcommand{\extract}{\algostyle{Extract}}

\newcommand{\wanthave}{\protstyle{WantHave}}
\newcommand{\block}{\protstyle{Block}}
\newcommand{\privwanthave}{\protstyle{PrivateWantHave}}
\newcommand{\privblock}{\protstyle{PrivateBlock}}

\newcommand{\pailliercrypto}{\protostyle{Paillier}}
\newcommand{\basicrlwe}{\protostyle{RLWEPIR}}
\newcommand{\basicpaillier}{\protostyle{PaillierPIR}}
\newcommand{\rlwe}{\protostyle{RLWE}}

\newcommand{\pir}{\protostyle{PIR}}

\newcommand{\rlwetwo}{\protostyle{RLWEPIR2}}
\newcommand{\rlwethree}{\protostyle{RLWEPIR3}}
\newcommand{\trivialpir}{\protostyle{TrivialPIR}}
\newcommand{\trivialspir}{\protostyle{TrivialSPIR}}

\newcommand{\ZZ}{\mathbb{Z}}

\newcommand{\NN}{\mathbb{N}}
\newcommand{\adv}{\mathcal{A}}
\newcommand{\qu}{\texttt{qu}}
\newcommand{\st}{\texttt{st}}
\newcommand{\db}{\texttt{db}}
\newcommand{\ans}{\texttt{ans}}
\newcommand{\sk}{\texttt{sk}}
\newcommand{\pk}{\texttt{pk}}
\newcommand{\ct}{\texttt{ct}}
\newcommand{\pt}{\texttt{p}}

\newcommand{\ceil}[1]{ \lceil #1 \rceil}

\newcommand{\pkeScheme}{\mathsf{PKE}}
\newcommand{\pirprot}{\mathsf{PIR}}
\newcommand{\KGen}{\mathrm{KGen}}
\newcommand{\Enc}{\mathrm{Enc}}
\newcommand{\Dec}{\mathrm{Dec}}
\newcommand{\indcpa}{\mathsf{indcpa}}
\newcommand{\pirsec}{\mathsf{pirsec}}
\newcommand{\spirsec}{\mathsf{spirsec}}
\renewcommand{\prf}{\mathsf{PRF}}
\newcommand{\advA}{\mathcal{A}}
\newcommand{\advB}{\mathcal{B}}
\newcommand{\Gm}{\mathsf{G}}
\newcommand{\Adv}{\mathsf{Adv}}

\newcommand{\cryptoenvlabelsize}{\scriptsize}

\newenvironment{code}[1][]{
\begin{enumerate}[label=\cryptoenvlabelsize\color{gray}\ttfamily\arabic*, ref=\arabic*, align=right, topsep=4pt, itemsep=0.05em, leftmargin=1.35em, #1]
}{
\end{enumerate}
}

\newenvironment{cralgorithm}[2][]{
\underline{#2:}
\begin{code}[#1]
}{
\end{code}
}

\ifusxversion
\newtheorem{theorem}{Theorem}[section]
\newtheorem{lemma}[theorem]{Lemma}
\newtheorem{definition}[theorem]{Definition}
\newtheorem{example}[theorem]{Example}
\fi

\makeatletter
\newbox\dottedarrow@box
\setbox\dottedarrow@box\hbox
  {%
    \begin{tikzpicture}
      \draw[dotted,->] (0,0) -- (1.2em,0);
    \end{tikzpicture}%
  }
\newcommand*\dottedarrow
  {\relax\ifmmode\expandafter\dottedarrow@m\else\expandafter\dottedarrow@t\fi}
\newcommand*\dottedarrow@t[1][1.5em]
  {\resizebox{#1}{!}{\raisebox{.5ex}{\usebox\dottedarrow@box}}}
\newcommand*\dottedarrow@m[1][]
  {%
    \if\relax\detokenize{#1}\relax
      \mathchoice
        {\dottedarrow@t}
        {\dottedarrow@t}
        {\dottedarrow@t[1.1em]}
        {\dottedarrow@t[0.9em]}%
    \else
      \dottedarrow@t[#1]%
    \fi
  }
\newbox\dottedleftarrow@box
\setbox\dottedleftarrow@box\hbox
  {%
    \begin{tikzpicture}
      \draw[dotted,<-] (0,0) -- (1.2em,0);
    \end{tikzpicture}%
  }
\newcommand*\dottedleftarrow
  {\relax\ifmmode\expandafter\dottedleftarrow@m\else\expandafter\dottedleftarrow@t\fi}
\newcommand*\dottedleftarrow@t[1][1.5em]
  {\resizebox{#1}{!}{\raisebox{.5ex}{\usebox\dottedleftarrow@box}}}
\newcommand*\dottedleftarrow@m[1][]
  {%
    \if\relax\detokenize{#1}\relax
      \mathchoice
        {\dottedleftarrow@t}
        {\dottedleftarrow@t}
        {\dottedleftarrow@t[1.1em]}
        {\dottedleftarrow@t[0.9em]}%
    \else
      \dottedleftarrow@t[#1]%
    \fi
  }
\makeatother

\begin{document}

\title{Peer2PIR: Private Queries for IPFS}

\ifacmversion
\else
\author{
{\rm Miti Mazmudar}\\
University of Waterloo
\and
{\rm Shannon Veitch}\\
ETH Zurich
\and
{\rm Rasoul Akhavan Mahdavi}\\
University of Waterloo
}
\fi

\ifacmversion
\else
\maketitle
\fi

\begin{abstract}
The InterPlanetary File System (IPFS) is a peer-to-peer network for storing data in a distributed file system, hosting over 190,000 peers spanning 152 countries.
Despite its prominence, the privacy properties that IPFS offers to peers are severely limited. 
Any query within the network leaks the queried content to other peers. 
We address IPFS' privacy leakage across three functionalities (peer routing, provider advertisements, and content retrieval), ultimately empowering peers to privately navigate and retrieve content in the network. 
Our work highlights and addresses novel challenges inherent to integrating PIR into distributed systems.
We present our new, private protocols and demonstrate that they incur reasonably low communication and computation overheads.
We also provide a systematic comparison of state-of-art PIR protocols in the context of distributed systems.
\end{abstract}

\section{Introduction}
Peer-to-peer (P2P) applications and networks 
have been widely used for decades and new networks, such as the InterPlanetary File System (IPFS), are being rapidly developed. 
IPFS is a P2P network hosting over 190,000 peers spanning 152 countries \cite{ipfsdesign}, making it one of the most widely used distributed file systems. 
The IPFS ecosystem~\cite{ipfs-ecosystem} includes many projects that facilitate building decentralized web applications. 
For example, Fleek~\cite{fleek} serves as a decentralized content distribution network backed by IPFS. 
Space Daemon~\cite{space-daemon} enables building distributed applications that process users' data privately on their own IPFS nodes. 

Although IPFS has proven itself to be an effective platform for a decentralized web, the current model provides limited privacy to users and lags far behind recent advancements in protecting Internet users. 
There has been significant adoption on the Internet of DNS-over-Encryption protocols, such as DNS-over-HTTPS and DNS-over-TLS \cite{rfc8484,edns,lu2019end}.
These alternatives to unencrypted DNS prevent malicious parties, ISPs, and others from observing the domains that are being accessed by Internet users.
Meanwhile in IPFS, data sent between two parties is encrypted, akin to Internet traffic being encrypted with TLS; however, any query within the IPFS network leaks to other peers the \emph{content for which a peer is querying}, akin to a DNS query. 
Encrypting messages between peers prevents passive observers from viewing information, but by revealing it to intermediate peers, a client's privacy is severely undermined.
Although distributed systems are lauded for their privacy guarantees, in reality, the current state-of-the-art  falls behind the privacy offered in the centralized setting.
Protecting the content of queries from intermediate peers is the natural next step to improve the privacy of IPFS, and provide users with end-to-end protection of their query content.

IPFS
provides three high-level functionalities to its peers. 
First, it arranges peers in an overlay topology such that each peer can 
contact and communicate with any other peer efficiently. 
IPFS implements this functionality of \emph{peer routing} using distributed hash tables (DHTs). 
Second, peers can advertise to other peers that they provide a file.
The \emph{provider advertising} functionality enables any peer on the network to discover which other peers provide a desired file. 
Third, a peer can directly contact another peer that provides a file and retrieve the file from them.
The two peers engage in a \emph{content retrieval} protocol to share that file. 
Each of these three functionalities is performed through protocols involving queries between different pairs of peers. 

Consider one peer acting as a client and another as a server in IPFS, where the server peer is storing content and assumed to be a passive adversary. 
In each of the three aforementioned functionalities, the client reveals what it is querying for to the server. 
Peer routing reveals which peer the client is attempting to contact.
In a provider advertisement query as well as a content retrieval query, a server learns which file the client wishes to retrieve.

Prior work addressed only one of these problems in isolation, e.g.,~by enabling confidentiality for the target file in content retrieval~\cite{bitswap-privacy,bitswap-privacy2} or obfuscating the target peer in peer routing~\cite{double-hashing}. 
However, replacing one aspect of the system with a private version 
is insufficient 
if the remainder of the system leaks the queried content. 
For example, a server peer can learn the target file through a provider routing query, even if a content retrieval query hides the target file.
Additionally, privacy for provider advertisements has not previously been addressed.

This work proposes an end-to-end solution for IPFS that hides the \emph{content a client queries for and retrieves in a response} from a server.
Ours is the first work to address 
this problem holistically across IPFS protocols for peer routing, provider advertising, and content retrieval.
This provides a comprehensive solution for advancing the privacy of users within IPFS --- not only is the content of their queries 
protected from intermediate peers, but also from the final peer from whom they collect the content.
Ultimately, no one within the network learns the content which a peer is retrieving.
In developing our solution, we handle challenges inherent to DHTs, such as a dynamic network with churn, and varying amounts of content stored at each peer.
Therefore, our resulting techniques are applicable to a wider range of distributed systems.
Towards developing this solution, we present the following building blocks:
\begin{itemize}[noitemsep, topsep=0pt]
    \item a novel routing algorithm that enables an end-to-end procedure for privately contacting a peer (\Cref{sec.privrouting});
    \item generic binning and hashing techniques to adapt datastores to be amenable to PIR, applied to provider advertisements and content retrieval (\Cref{sec.prov-ads,sec.privbitswap});
    \item a scalable and private adaptation of an interactive content retrieval protocol (\Cref{sec.privbitswap}); and,
    \item a systematic comparison of state-of-the-art PIR protocols in the context of distributed systems (\Cref{sec:pir-scheme-analysis}).
\end{itemize}

Furthermore, we find that for our use cases in DHTs, existing PIR protocols do not provide a reasonable communication-computation tradeoff.
To address this gap, we present new PIR protocols, which combine techniques from cryptographic folklore and PIR literature (\Cref{sec.basicpir}), and are optimized particularly for small databases.
We prove the security of our PIR protocols and demonstrate that they outperform the state-of-the-art in our setting.

Our performance evaluation demonstrates the feasibility of our system in terms of communication and computation costs.
Our private peer routing and provider advertisement queries can be answered in less than 100ms and 1.5s, respectively.
We believe that these are reasonable latencies for the IPFS system, given the benefit of added privacy.
Furthermore, our private protocols do not incur any additional rounds and our private peer routing protocol does not significantly increase the number of hops over the original routing algorithm.
Our system is modular in that it enables substituting PIR protocols for each of the three functionalities, and thus Peer2PIR directly benefits as PIR protocols become more efficient.
We integrate our protocols into the IPFS codebase, highlighting and addressing privacy engineering challenges faced along the way.

\section{Background}\label{sec.background}

\begin{figure*}[t]
\begin{tikzpicture}[
round/.style={circle, draw=green!60, fill=green!5, very thick, minimum size=7mm},
square/.style={rectangle, draw=red!60, fill=red!5, very thick, minimum size=7mm},
hexagon/.style={regular polygon, regular polygon sides=6, align=center, draw=red!60, fill=red!5, very thick, minimum size=5mm},
offline/.style={dashed},
circled/.style={circle, draw, inner sep=2pt},
round2/.style={circle, draw=blue!60, fill=blue!5, very thick, minimum size=7mm},
round3/.style={circle, draw=red!60, fill=red!5, very thick, minimum size=7mm},
]
\pgfkeys{/csteps/fill color=black}
\pgfkeys{/csteps/inner color=white}

\draw(-6,-1.5) node[round2, label=above left:{\color{blue!60}Wants CID $C$}] (client)   {Client};

\draw(6,1.5) node[round3, label=below right:{\color{red!60} Holds Content $C$}] (target)     {P4};
\draw(3,1.5) node[round]      (P3)    {P3};
\draw(-3,1.5) node[round]      (P1)    {P1};
\draw(0,1.5) node[round]      (P2)    {P2};

\draw[->,dotted,line width=1pt] 
    ([xshift=-8pt]client.north) to [bend left=35]  
        node[pos=0.5,sloped, above=1pt] {\small{\Circled{\textbf{1}}} \& \small{\Circled{\textbf{2}}}} 
    ([yshift=7pt,xshift=2pt]P1.west);


    
\draw[<-,dotted,line width=1pt] 
    (client.north) to [bend left=35]
    node[midway, sloped, below=1pt] {\small{\Circled{\textbf{3}}} P2 \& \small{\Circled{\textbf{4}}}}
    (P1.west);

\draw[->,dotted,line width=1pt] 
    ([yshift=4pt,xshift=-2pt]client.north east) to [bend left=5]  
    node[pos=0.8, sloped, above=1pt] {\small{ \Circled{\textbf{5}}} \& \small{\Circled{\textbf{6}}}} 
    ([yshift=1pt]P2.west); 


\draw[<-,dotted,line width=1pt] 
    (client.north east) to [bend left=5]  
    node[pos=0.75,sloped, below=1pt] {\small{\Circled{\textbf{7}}} P3 \& \small{\Circled{\textbf{8}}}} 
    ([yshift=-3pt,xshift=1pt]P2.west);

\draw[->,dotted,line width=1pt] 
    ([yshift=4pt,xshift=-2pt]client.east) to [bend right=15]  
    node[pos=0.8,sloped, above=1pt]  {\small{\Circled{\textbf{9}}} \& \small{ \Circled{\textbf{10}}}}
    (P3.west);



\draw[<-,dotted,line width=1pt] 
    (client.east) to [bend right=15]  
    node[pos=0.7,sloped, below=1pt] {\small{\Circled{\textbf{11}}} \& \small{\Circled{\textbf{12}}} P4} 
    ([xshift=-2pt]P3.south west) ; 

\draw[->,line width=1pt] 
    ([xshift=2pt,yshift=2pt]client.south east) to [bend right=25]
    node[pos=0.7,sloped,above=1pt]{\small{\Circled{\textbf{13}}}}
    (target.south west);
    
\draw[<-,line width=1pt] 
    ([xshift=-4pt,yshift=-1pt]client.south east) to [bend right=26]
    node[pos=0.7,sloped,below=1pt]{\small{\Circled{\textbf{14}}}}
    ([xshift=-4pt]target.south); 

\node(legend-text) [draw=gray!60, rounded corners,
                    fill=gray!8,
                    text width=2\columnwidth,
                    align=flush center, 
                    inner sep=12 pt]%
    at (0,4.2)
    {\begin{tabular}{llll}
        & & \faLockOpen\ \textbf{Standard IPFS} & \faLock\ \textbf{Ours (Peer2PIR)} \\[4pt]
    \multirow{2}{*}{$\dottedarrow$} & \small{\Circled{\textbf{1}} \Circled{\textbf{5}} \Circled{\textbf{9}}} Peer routing request & \texttt{Peers near Hash(C)?} & \texttt{PrivPeerRoutingQuery(C)}\\
    & \small{\Circled{\textbf{2}} \Circled{\textbf{6}} \Circled{\textbf{10}}} Prov. advert. request & \texttt{Provs. for C?} & \texttt{PrivProvAdQuery(C)} \\[4pt]
    \multirow{2}{*}{$\dottedleftarrow$} & \small{\Circled{\textbf{3}} \Circled{\textbf{7}} \Circled{\textbf{11}}} Peer routing response & \texttt{PID} & \texttt{PrivPeerRoutingResp(PID)} \\
    & \small{\Circled{\textbf{4}} \Circled{\textbf{8}} \Circled{\textbf{12}}} Prov. advert. response & \texttt{PID/None} & \texttt{PrivProvAdResp(PID/None)} \\[4pt]

    \multirow{2}{*}{$\longrightarrow$} &  \small{\Circled{\textbf{13}}} Content retrieval request & \texttt{WantHave C?} / \texttt{WantBlock C} & \texttt{PrivWantHave(C)} \\
    & \small{\Circled{\textbf{14}}} Content retrieval response& \texttt{Have} / \texttt{Block C} & \texttt{PrivBlock(C)}
    \end{tabular} 
    };

\end{tikzpicture}
\caption{Overview of the IPFS peer routing, provider routing, and content retrieval protocols, and our private versions. }
\label{fig:IPFS-overview}
\end{figure*}

\subsection{Distributed Hash Tables \& IPFS}
\label{sec.background-dhts-ipfs-architecture}
Distributed hash tables (DHTs), such as Kademlia~\cite{kademlia} and Chord~\cite{chord}, are analogous to conventional hash tables, with the exception that the entire key-value store is systematically split across peers in the P2P network so that queries can be conducted efficiently. 
Importantly, both the content stored for the application as well as the routing information for the P2P network, are distributed across peers. 
P2P networks grow dynamically --- the rate at which peers leave and rejoin the network is the \emph{node churn} rate. 
Though each peer can respond to and conduct queries, throughout this paper we designate peers as \emph{clients} when they are conducting queries, and as \emph{servers} when they are responding to queries.
IPFS is a distributed file system built on the Kademlia DHT~\cite{ipfsdesign,ipfs-docs}. 
Next, we describe the three main functionalities provided by IPFS: contacting another peer in the DHT, discovering which peers provide desired content, and retrieving content from a peer.
We visualize these functionalities, alongside our private alternatives, in \Cref{fig:IPFS-overview}. 

\paragraph{Contacting a Peer.}
Each peer in a DHT has a peer identifier (peer ID) that is the output of a cryptographic hash function. For example, in IPFS, peer IDs are given by a hash of the peer's public key. 
Each peer in the DHT maintains a \emph{routing table} consisting of peer IDs and their routing information, for a small number of other peers in the DHT. 
The routing information consists of \emph{multiaddresses} which encode IP addresses and ports, as well as peer IDs. 
To reach a target peer, a client obtains relevant routing information by querying one of its neighbours for the target peer ID, 
and then querying one of that neighbour's neighbours, and so on, through an iterative routing process. 
DHTs guarantee that each peer can reach any other peer in a network of size $n$, without a trusted central party, through approximately $\log(n)$ routing queries.
Importantly, each such routing query will reveal the target peer ID to the in-path peers.
The IPFS routing algorithm in the context of the Kademlia DHT is discussed in \Cref{subsec:kademlia-and-ipfs-routing-table}.

\paragraph{Discovering Peers Providing Content.}     %
IPFS decomposes each file into multiple content blocks, using either fixed or adaptively sized chunks~\cite{ipfs-chunking,ipfs-chunking-masters-thesis}.
Each content block in IPFS is addressed by a \emph{content identifier} (CID), which includes a collision-resistant one-way hash of the block. 
CIDs remove the need for a central party to address content and are also used to verify the integrity of a content block.   
A peer who provides a content block in IPFS, namely a provider peer, also advertises this fact to the DHT, through provider advertisements. 
A provider advertisement maps a given CID to its providers' peer IDs. 
    Each peer maintains a provider store with provider advertisements. 
    So a client who wishes to discover the peers who provide a given CID, should query the server peers' provider store for the target CID. 
    (This query occurs simultaneously with the aforementioned routing query, as we describe later.)
    Evidently, through this query, the server peer learns the target CID that the client wishes to access. 
    We describe IPFS provider advertisements in \Cref{subsec:ipfs-provider-ads}.

\paragraph{Retrieving a Content Block.}
Once a client determines the provider peer for a target CID, it proceeds to retrieve the content block. 
Content retrieval in IPFS is accomplished using Bitswap~\cite{bitswap}, a content exchange protocol for retrieving a content block with a given CID 
from a connected peer. 
We detail the Bitswap protocol in \Cref{subsec:bitswap-protocol} and note that the server learns the target CID that the client wishes to fetch as well as the content block that it returns.

\subsection{Private Information Retrieval}
\label{sec.background-pir}
Private information retrieval (PIR) protocols allow a client to retrieve an item from a database, held by one or more servers, such that the servers do not learn which item is retrieved.
Information-theoretic PIR (IT-PIR) schemes provide perfect secrecy, even in the presence of computationally unbounded adversaries.
Many efficient IT-PIR schemes exist in the multi-server setting~\cite{chor1995private,Beimel2002BreakingTO,raidpir,Ambainis2000};
however, they rely on assumptions that are difficult to achieve in practice, as we discuss in~\Cref{sec.discussion}. 
In contrast to IT-PIR, computational PIR (CPIR) schemes rely on cryptographic hardness assumptions and assume a computationally bounded adversary.
We mainly focus on CPIR schemes for our protocols.
We defer the formal definition of a PIR protocol to \Cref{app:pir-definition}. 
While PIR schemes ensure that the server does not learn which record is retrieved, symmetric PIR (SPIR) schemes also guarantee that the client learns only the record for which they query. 
SPIR is similar to oblivious transfer (OT).

\paragraph{IndexPIR vs. KeywordPIR.}
PIR schemes often assume the client knows the index of the desired row in a table; these are known as IndexPIR schemes.
KeywordPIR schemes enable a client to retrieve a value from a key-value store, given that they know the key.
KeywordPIR schemes privately match keywords, and thus incur communication overhead over IndexPIR schemes~\cite{cpir, pantheon}.
PIR over key-value stores can be performed using an IndexPIR scheme,
with an extra round of communication and a small probability of error.
The server constructs a hash table, mapping each key to an index in the table.
It places the associated value into the indexed row of the hash table and sends the chosen hash function to the client.
The client, using the hash function and its desired keyword, calculates the index for its desired row in the hash table and then uses IndexPIR to retrieve that row.
As modern DHTs involve key-value stores in their design, we use this transform to efficiently conduct IndexPIR. 
We exploit the fact that the CIDs themselves include hashes, to enable the client to compute the index for a CID (within the server's hash table), without an extra round.

\paragraph{PIR in Distributed Settings.}
The DHT setting presents challenges in immediately applying PIR.
First, the key-value stores have varying numbers of records and record sizes. 
While the routing table is upper-bounded in size, the number of provider advertisements and content blocks on each peer follows a long-tail distribution, as we discuss in \Cref{sec.threatmodel}.
Since PIR schemes are typically optimized for one database size, one scheme does not fit all three of our use cases. 
Second, we seek to minimize the communication overhead of PIR queries, to prevent congesting the distributed network, while maintaining reasonable computational overheads for individual peers.
Furthermore, existing schemes commonly assume that clients repeatedly query the same server, amortizing communication cost over many queries.
However, across all types of queries that we consider, a peer acting as a client routinely contacts different peers.
Thus, a direct application of state-of-the-art PIR protocols optimized for communication overhead is not possible for our use cases.

\section{Related Work}\label{sec.relatedwork}

\paragraph{Private Routing and Provider Advertisements.} 
The recently proposed Double Hashing \cite{double-hashing,double-hashing-video} aims to improve client privacy when querying for content in IPFS. 
The approach requires that a client query for a \emph{prefix} of the CID as opposed to the entire CID, thus providing the client with plausible deniability about which content it is querying for. 
A server returns provider advertisements corresponding to any CID matching this prefix. 
The client thus obtains $k$-anonymity where $k$ is the number of CIDs with a matching prefix. 
We provide stronger guarantees to both the client and the server. 
The client hides the entire CID, instead of just the suffix, from the server. 
The server is guaranteed that the client does not learn other CIDs stored by the server. 

\paragraph{Private Content Retrieval.}
Daniel et al.~\cite{bitswap-privacy} provide plausible deniability to peers by integrating a gossip protocol into Bitswap, in which peers forward messages on each other's behalf. 
While their work obscures the source of a query, we protect the content of a query.
Daniel and Tschorsch~\cite{bitswap-privacy2} present two protocols to hide the target CID from the server during the content \emph{discovery} step of Bitswap. 
However, the first protocol leaks the entire list of CIDs to the client, whereas the second protocol relies on computationally expensive Private Set Intersection (PSI) operations. 
Both protocols reveal the target CID to the server during the content \emph{retrieval} step. 
Mazmudar et al.~\cite{dhtpir} incorporate IT-PIR into DHTs to hide the content that clients retrieve within the network. 
This approach relies on grouping peers into \emph{quorums} which contain a bounded fraction of malicious nodes. 
Under this assumption, their construction can take advantage of existing protocols~\cite{rcp, qp} for robust DHTs, enabling integrity and confidentiality of routing queries within the DHT. 
We discuss challenges with incorporating similar strategies into IPFS in \Cref{sec.discussion}.
Keyword search schemes for IPFS enable a client to retrieve a file, given a representative keyword rather than its CID~\cite{cao_and_li}.
They also reveal the keyword to server peers. 

\section{Threat Model}\label{sec.threatmodel}
We consider a peer acting as a client and an adversarial peer acting as a server. 
We aim to hide the \emph{content that the client queries for and retrieves} from the server(s), across three query types: 

\begin{itemize}[topsep=0pt]\itemsep0mm
    \item Peer routing queries: the client fetches routing information for a \emph{target peer ID} from a server.
    \item Provider advertisement queries: the client fetches which peers provide a \emph{target CID}, as well as the routing information to reach these provider peers. 
    \item Content retrieval queries: the client fetches a block from a server peer, given its \emph{target CID}. 
\end{itemize}
We hide the target peer ID from each intermediate server peer during peer routing queries. We hide the CIDs during provider advertisement queries and content retrieval queries. 
Additionally, the server peer cannot determine these identifiers based on its response to a client's query.
We strictly improve over the current IPFS implementation, by hiding these identifiers, as IPFS currently only encrypts communication between peers at the transport layer~\cite{ipfs-privacy-encryption}. 

We assume that all peers within the network are honest but curious, i.e., they correctly follow prescribed protocols but may passively observe their execution to infer information. 
We briefly outline non-goals before contextualizing our guarantees for provider advertisements and content retrieval.

\paragraph{Non-Goals.}
We do not aim to protect the client's network identity. 
Instantiating a distributed file system such as IPFS over Tor onion services to provide client anonymity is outside of our scope.
We do not protect against a network adversary that observes the client. 
Based on the volume of (encrypted) traffic to other IPFS peers, such an adversary can narrow down the search space of possible target peer IDs or CIDs that the client is looking for. 
Similarly, peers who are acting as servers cannot collude to share the fact that they received a request from a given client (peer ID), as they can then conduct this attack.
We do not protect content \emph{publishing} requests.
In other words, peers know that a server peer is storing a given CID (via provider advertisements or by simply querying the server for the content block).
Since CPIR schemes add significantly more computational overhead for the server than the client, a malicious client can exploit our CPIR-based schemes to amplify a denial of service attack against any peer. 
Standard rate-limiting based on the client's peer ID and network identity can be used to deter such DoS attacks. 

\paragraph{Peer Routing.} 
Servers may cause an honest client to experience a restricted view of the network through \emph{eclipse attacks}. 
A server peer may respond incorrectly to a client's routing request, causing it to be directed to other malicious server peers. 
Pr\"{u}nster et al.~\cite{teoth} demonstrated an eclipse attack against IPFS, which has since been fixed. 
Importantly, the client's routing table serves as feedback to a server peer in conducting their attack.
While we do not prevent eclipse attacks, we do not make IPFS peers more vulnerable to these attacks, e.g.~by releasing a client's routing table.

\paragraph{Provider Advertisements and Content Retrieval.}
We ensure that a client can only retrieve a provider advertisement or a content block if it knows the corresponding CID. 
In other words, the client should not be able to retrieve information stored by a server peer corresponding to CIDs that are unknown to the client. 
Separately, we observe that 
the privacy guarantees afforded to a client depend on the amount of content stored by server peers. 
The amount of provider advertisements stored by each peer in IPFS follows a long tail distribution~\cite{dennis-providers-study}. 
So it is highly likely that a client will query a heavyweight content provider for both provider advertisements and content blocks, and its target CID will be indistinguishable from millions of others.
Moreover, assuming that a peer stores files of different sizes, the number of content retrieval queries made by the client may leak which file is being retrieved. 
We cannot prevent this attack without changing how content blocks are duplicated in IPFS.

\section{Private Peer Routing}\label{sec.privrouting}
In this section, we focus on how a server answers a client's query for a target peer ID without revealing the peer ID to the server. 
We first outline the Kademlia DHT routing process, highlighting how a server uses its routing table.
We then propose an algorithm to adapt the routing table to enable hiding target peer IDs from the server and demonstrate its correctness. 
Finally, we integrate a PIR scheme to privately retrieve a bucket from our modified routing table, leading to a private peer routing algorithm. 

\subsection{Kademlia and the IPFS Routing Table} 
\label{subsec:kademlia-and-ipfs-routing-table}

    \paragraph{Kademlia DHT.}
    DHTs map both the content identifiers and peer identifiers to a common virtual address space, through a collision-resistant one-way hash. 
    IPFS uses \sha{} to form the virtual address space for its underlying Kademlia DHT. 
    Kademlia measures \emph{distance} between two addresses by computing an XOR between them and interpreting the result as an integer.
    For example, the addresses \texttt{0b00011} and \texttt{0b00101}, in an example 32-bit address space, are \texttt{0b00110=6} units apart. 
    They share the first 2 bits and are said to have a \emph{common prefix length} (CPL) of 2. 
    Addresses with higher CPLs are said to be closer to one another. 

    \paragraph{Routing Table Structure.}
    The routing table (RT) is indexed by CPLs. 
    When a peer is added to an RT, the peer first computes the CPL between itself and the new peer. 
    Each row of the RT contains multiaddresses for at most $k$ peers whose CPL is given by the row index. %
    In other words, higher indices of the RT contain closer peers. 
    In IPFS, each row (also called \emph{bucket}) can have up to $k = 20$ peers' multiaddresses. 
    When a peer first joins the network, it contacts a set of bootstrapping peers to populate its RT.
    Whenever a peer comes into contact with another peer in the DHT, it may create a new bucket for this peer or fill up an existing bucket, depending on the CPL between the two peers.
    Buckets are created dynamically, so a bucket with index $i$ is created only when the peer has met at least $k$ additional peers whose CPL is greater than or equal to $i$.
    If an existing bucket is full, then the new peer is not stored in the RT, leading to Kademlia preferring longer-lived peers.
    In general, the RT can have buckets of unequal sizes.

    \paragraph{Kademlia Lookup.}
    We focus on how a server uses its RT to answer a client's query for the target peer ID. 
    If the bucket in the RT given by the CPL between the target and the server's peer IDs is full, then the server returns that bucket. 
    Otherwise, it selects peers from other buckets to return a set of $k$ closest peers.
    It does so by computing the distance between the target peer ID and the remaining peer IDs in the RT to select the $k$ closest peers.
    As long as the RT has at least $k$ peers in total, for any lookup, the server returns the multiaddresses for $k$ peers. 
    The client then uses the multiaddresses to contact these peers, which are closer to the target peer ID than the server peer. 
    The client iteratively queries different server peers until it reaches the target peer. 

    \paragraph{Challenges in Hiding the Target Peer ID.}
    Our goal is to enable the server to respond to the client's query, while hiding the target peer ID.
    Suppose that the server only learns the bucket index, namely the CPL between its own peer ID and the target peer ID. 
    (We ultimately hide this 
    index as well, via PIR, as we discuss later.)
    If the server's bucket for this CPL is full, then it simply returns this bucket.
    However, if it is not full, then only knowing the CPL is insufficient to complete the original Kademlia lookup.
    In particular, the suffix of the target peer ID is used to compute distances to peers in other buckets, and select the closest ones.
    A naive approach might return the bucket corresponding to the CPL, regardless of if it is full; however, non-full buckets would be insufficient for the client to reach the target peer, and the bucket size could leak its index.  
    To solve the problem of non-full buckets, the server can compute the $k$-closest peers for every target peer ID.
    This approach blows up the RT, which originally has at most $256$ rows, to $2^{256}$ rows. 
    
    We propose \emph{normalizing} the RT, that is, ensuring that each bucket has exactly $k$ peers, which are closest to any target peer ID whose CPL is given by the bucket index. 
    We present an algorithm for normalizing the RT buckets. 
    We introduce a small amount of randomization when selecting closest peers while knowing only the target peer's CPL. 
    We demonstrate that our algorithm introduces only a constant overhead to the number of hops that the client must perform, for any target peer ID, in comparison to the original Kademlia lookup.
    In practice, as we show next, our algorithm introduces negligible to no overhead.        

\begin{algorithm}[!ht]
\caption{Routing Table Normalization Algorithm}\label{alg:norm-rt}
\begin{algorithmic}[1]
    \State $i \in [0,r = \mbox{len}(\mbox{RT})]$, bucket $B_i$ corresponds to CPL $i$. \\
    \textbf{Input:} Target index $t$
    \State $R \gets \cup_{i=0}^{r} B_i$
    \If{$\|R\| \leq k$} \Comment{RT has less than $k$ peers}
        \Return $R$  \label{alg:less-than-k-total}
    \EndIf
    \If{$r < t$} \Comment{If $r < t$, return last bucket} 
        \State $t \gets r$ \label{alg:return-last-bucket}
    \EndIf
    \State $R \gets B_t$
    \If{$\|R\| == k$} 
        \Return $R$ \label{alg:bucket-is-full}
    \Else \Comment{$R$ is not full ($\|R\| < k$)}
        \If{$t < r$} \Comment{Closer buckets can be used} \label{alg:prefer-closer-buckets} 
            \State $C \gets \cup_{i=t+1}^{r} B_i$ 
            \If{$\|C\| \leq k - \|R\|$} \label{alg:closer-buckets-are-not-sufficient}
                \State $R \gets R \cup C$ 
            \Else\ append $k - \|R\|$ peers from $C$ to $R$ at random \label{alg:closer-buckets-pick-at-random}
            \EndIf
        \EndIf
        \If{$\|R\| < k$} \Comment{Must go through farther buckets}
            \State $\ell \gets t-1$
            \While{$\|R\| + \|B_\ell\| \leq k$} \label{alg:pick-each-farther-bucket-completely-start}
                \State $R \gets R \cup B_\ell$
                \State $\ell \gets \ell - 1$ \label{alg:pick-each-farther-bucket-completely-end}
            \EndWhile
            \State $B_{\ell} \gets \{ B_{\ell}^{j} \}$ \Comment{Sort peers in $B_{\ell}$ by distance to server (smaller $j$ = closer to server)} \label{alg:subbuckets}
            \State $m \gets 1$
            \While{$\|R\| + \|B_\ell^m\| \leq k$} \label{alg:pick-subbuckets-start}
                \State $R \gets R \cup B_\ell^m$
                \State $m \gets m + 1$ \label{alg:pick-subbuckets-end}
            \EndWhile
            \State append $k - \|R\|$ peers from $B_\ell^m$ to $R$, at random. \label{alg:pick-at-random-from-subbucket}
        \EndIf
    \EndIf
    \Return $R$
\end{algorithmic}    
\end{algorithm}

\subsection{Routing Table with Normalized Buckets}
\label{subsec:normalization}
We present our normalization algorithm (\Cref{alg:norm-rt}) and describe it next. 
The client wishes to retrieve a target bucket indexed by $t$, where $t$ is the CPL between the target and the server's peer IDs.
As in the original lookup algorithm, if the RT has at most $k$ peers \emph{in total}, all peers are returned (Line~\ref{alg:less-than-k-total}). 
Since buckets are created dynamically, the RT may have $r < t$ buckets, in which case the server returns the (normalized) last bucket, containing the closest $k$ peers that it knows (Line~\ref{alg:return-last-bucket}).
If the target bucket is full, it is returned directly (Line~\ref{alg:bucket-is-full}).

We observe that a peer in a bucket $t$ would share the common prefix of length $t$ not only with the server peer, but also with peers in closer buckets ($t+1\leq i \leq r$).
Whereas its common prefix with peers in farther buckets ($0\leq i \leq t-1$) is given by the index of the farther bucket.  
In particular, a peer in a bucket $t$ is closer to peers in buckets that are closer to the server and farther from peers in buckets that are farther from the server. 
Thus the server prefers filling a bucket with peers from closer buckets over farther ones (Line~\ref{alg:prefer-closer-buckets}). 
Importantly, without knowing the suffix of the target peer ID, the server cannot order peers in these closer buckets in terms of their distance to the target.
So, the server picks the remaining number of peers at random from all closer buckets, 
taking all closer peers if there are less than $k$ such peers
(Lines~\ref{alg:closer-buckets-are-not-sufficient}--\ref{alg:closer-buckets-pick-at-random}). 


If bucket $t$ remains non-full, the server then adds farther buckets one at a time (Lines~\ref{alg:pick-each-farther-bucket-completely-start}--\ref{alg:pick-each-farther-bucket-completely-end}).
If the server cannot add a bucket completely, then 
it partitions the bucket into equivalent subbuckets, based on its distance to them
(Line~\ref{alg:subbuckets}), as illustrated in Example~\ref{eg:normalization-subbuckets}.
The server then adds each subbucket one at a time (Line~\ref{alg:pick-subbuckets-start}--\ref{alg:pick-subbuckets-end}), selecting peers at random if it cannot include the entire subbucket (Line~\ref{alg:pick-at-random-from-subbucket}).

\begin{example}\label{eg:normalization-subbuckets}
A server with peer ID \texttt{0b0011} has 2 buckets:
\begin{itemize}
    \item $B_2$ with the common prefix \texttt{0b00XX}. It can contain peers with IDs \texttt{0b001X}. It currently has \texttt{0b0010}.
    \item $B_1$ with the common prefix \texttt{0b0XXX}. It can contain peers with IDs \texttt{0b01XX}. It currently has \texttt{0b0101}, \texttt{0b0111}, \texttt{0b0110}. 
\end{itemize}
Suppose we wish to normalize $B_2$ to $k=3$ peers.
Peers in $B_1$ that are closest to $B_2$ would share the later bit(s) of the common prefix, i.e. they would be of the form \texttt{0b011X}, so that $\texttt{0b001X} \oplus \texttt{0b011X}$ is minimized to $[4,6)$. 
We can partition $B_1$ into: $B^{1}_{1}=\texttt{0b011X}=\{\texttt{0b0111}, \texttt{0b0110} \}$ and $B^{2}_{1}=\texttt{0b010X}=\{\texttt{0b0101}\}$. Then, a normalized $B_2$ would contain $\{\texttt{0b0010}, \texttt{0b0111}, \texttt{0b0110}\}$.
\end{example}

\paragraph{Convergence.}
\label{sec:proof-of-convergence}
We prove and experimentally verify that our normalized algorithm does not significantly increase the number of hops required to reach a target peer. 
At a high level, the argument follows from the fact that the only cases in which the normalized algorithm diverges from the original still provide the \emph{same guarantee} of closeness to the target peer.
In the lemma, \ipfskadem{} refers to the original routing algorithm and \privipfskadem{} refers to our normalized routing algorithm.
We defer the proof to \Cref{app.proof}. 

\begin{lemma}[Convergence of Private Routing]\label{thm.normal}
If \ipfskadem{} converges in $\lceil \log n \rceil + c_1$ hops, then \privipfskadem{} converges in $\lceil \log n \rceil + c_2$ hops for some constants $c_1, c_2$.
\end{lemma}

We experimentally evaluated our normalized routing algorithm against the Trie-based routing algorithm currently deployed in IPFS to determine the difference in the number of hops required to reach a target peer.
We simulated a network that mirrors the topology of IPFS and queried for different target peer IDs, tracking the number of hops required by each algorithm.
Our experiments verify Lemma \ref{thm.normal} and show that in practice, normalizing the routing table introduces negligible to no difference in the number of hops. 

To capture a meaningful network topology to test our claim, we crawled the active IPFS network using the Nebula crawler \cite{nebula} and collected an adjacency list of all dialable peers and their neighbours. 
Our crawl collected 13692 peer IDs from the network; we note that no peers' network addresses were gathered.
For each peer, we set up a Trie-based RT using the current Kademlia implementation and a normalized RT based on our new algorithm. 
We instantiated our experiment with the same parameters as in IPFS: for each step of the routing process, the 
client concurrently queries three closest peers it knows, to the target peer ID, and retrieves $k=20$ nearest peers from them. 
The process is repeated iteratively until the target peer is found. 
We tested 5000 target peer IDs at random, from a constant client peer. 

All queries took an average of two hops to reach the target, with a maximum of six hops.
We did not find a difference in the number of hops using the Trie-based and the normalized routing algorithms, that is, both algorithms required an identical number of hops in every iteration of our test.
Thus, our normalized routing algorithm imposes negligible overhead over the original routing algorithm.

\subsection{Private Algorithm Integration}
\label{sec:private-routing-pir-integration}
    
    Integrating our normalization algorithm into IPFS generates privacy engineering challenges. 
    For instance, the contents of the RT in practice differ from the Kademlia model. 
    We discuss why PIR is optimal for privately retrieving buckets from the RT, and then integrate our algorithm with PIR to develop a private peer routing solution. 

    \paragraph{Dynamic Networks.}
    IPFS is a dynamic network with high churn, meaning that nodes often join and leave the network.   
    Consequently, clients' RTs are regularly updated in order to reflect new or failing peers. 
    We require normalizing the RT every time a new peer is added to it.
    Since routing tables are relatively small and the normalization algorithm is not computationally intensive, this does not impose any prohibitive overhead on the clients. 
    We discuss the impact of node churn on our private peer routing in~\Cref{app:churn}. 
    
    \paragraph{Joining Peer IDs and Addresses.}
    We have thus far assumed that our RT stores multiaddresses to peers. In practice, the RT in IPFS stores only the \emph{peer IDs}, and an additional table acting as an \emph{address book} maps these peer IDs to their multiaddresses. 
    This allows multiaddresses to be stored only once and used for both peer and provider routing. 
    A non-private lookup uses the value fetched from the RT as a key to look up the address book;
    however, such a lookup is not possible when privately retrieving buckets from the RT.
    Therefore, in our implementation, the server joins the two key-value stores, namely the RT which stores peer IDs and address book which stores multiaddresses, before responding to a private query. 
    This enables PIR to be performed over one table rather than two. 
    The join must be recomputed whenever either store is updated, which is a function of the node churn rate. 

    \paragraph{Selecting a PIR scheme.}
    Following normalization and join steps, we must hide the target CPL from the server. 
    The RT is limited in both its number of rows ($\leq$ 256) and the size of each row (multi-addresses for $k=20$ peers).
    Under a trivial PIR approach, the server would provide the entire RT to the client, foregoing normalization and joining;
    however, trivial PIR makes the server susceptible to eclipse attacks
    ~\cite{teoth}. 

    Alternatively, we can use OT, which would restrict the client to retrieving only one row and prevent leaking other rows' contents.
    However, 
    efficient OT algorithms typically involve encrypting each row of a table with a different key and then obliviously transferring a particular cryptographic key to a client, therefore making the total communication cost greater than the size of the table. 
    This introduces a trade-off between information leaked to clients about other rows in the routing table and communication cost. 
    Given the fact that IPFS does not restrict a client from repeatedly querying a server for different target peer IDs, i.e., rows in the routing table,
    we determine that the minimal leakage guarantee provided by OT does not justify the communication overhead.  
    Rate-limiting on the server peer can reduce the susceptibility of server peers to eclipse attacks.
    We thus opt for PIR with rate-limiting over OT. 
    \iffullversion
    We discuss client-side implementation details in \Cref{app:implementation-details}.
    \fi
    
    We design an IndexPIR scheme based on RLWE, namely~\basicrlwe, catered to our private routing use case.
    This scheme and design rationale are detailed in \Cref{sec:pir-scheme-analysis}.

    \paragraph{PIR Integration.}
    A client first computes the CPL between the target peer ID and the server peer ID, which it uses to construct a PIR query.
    The server uses the PIR query and its normalized routing table joined with multiaddresses, to compute a PIR response, which is returned to the client.
    The client processes the PIR response, to obtain multiaddresses of $k$ peer IDs close to the target peer ID.  
    Effectively, the use of PIR hides the target CPL from the server.


    \paragraph{Security.}
Interactions between the client and the server now consist only of PIR queries and responses, so security follows directly from that of $\basicrlwe{}$, as in \Cref{sec:pir-scheme-analysis}.

\section{Private Provider Advertisements}
\label{sec.prov-ads}
We outline the provider advertisement functionality in IPFS, i.e., the method by which clients determine which peers are providing desired content, and introduce our private alternative.

\paragraph{Notation.}\label{sec.notation}
The following notation is useful for the next two sections.
Each peer stores $m$ content blocks, each having a unique CID.
We refer generically to a CID as $c$.
A peer holds an ordered list $L$ of CIDs of the content blocks they store.
We refer to the $i$-th CID by $L[i]$. 
We refer to the first $\ell$ bits of the CID by the notation $[:\ell]$ and so $L[i][:\ell]$ refers to an $\ell$-bit prefix of the $i$-th CID in $L$.
We form a table $D$ consisting of all content blocks, 
ordered by their CIDs, such that $D[i]$ is the content block for the $i$-th CID in $L$.
Each peer also maintains a provider store $P$, which maps each CID for which the peer knows a provider, to the provider peers' IDs. 
The address book $M$ maps a peer ID to its multiaddress. 
$\mbox{SE}.\encrypt(k,x)$ refers to the (robust) symmetric encryption of $x$ with the key $k$. For example, this can be instantiated with AES-AEZ~\cite{aes-aez}.
We refer to the homomorphic encryption (HE) of $x$ as $\mbox{HE}.\encrypt(x)$. 
Let $\mbox{KDF}$ denote a \emph{key derivation function}, which takes as input some key material and outputs an appropriate cryptographic key. Implicitly, the $\mbox{KDF}$ outputs a key of the desired length (in our case, the length of key required for the scheme $\mbox{SE}$).

\subsection{IPFS Provider Advertisements}
\label{subsec:ipfs-provider-ads}
A provider advertisement maps a CID $c$ to the peer ID $p$ of a peer that provides $c$, namely a provider peer. 
Each peer maintains a key-value store of provider advertisements, known as a provider store. 
When a provider peer stores a content block addressed by $c$, it advertises the block to the IPFS network as follows. 
The provider peer contacts its $k$-closest neighbours to $c$, following the aforementioned iterative routing process. 
(A routing query for a CID operates similarly to that for a peer ID, in that both identifiers get hashed to the \sha{} virtual address space of the Kademlia DHT, as observed in Section~\ref{subsec:kademlia-and-ipfs-routing-table}.) 
It then requests each of these neighbours to add a provider advertisement for this CID to their provider store.

Consider a client who wishes to discover the provider peers for $c$. 
The client first runs a routing query for that CID. 
This routing query will ultimately lead the client to one of the $k$-closest peers to $c$.
We note that our private peer routing algorithm from Section~\ref{sec:private-routing-pir-integration} can be used to determine $k$-closest peers to the desired CID, without revealing the CID to the server peer. 
As described above, these peers will store the provider advertisement, which would map CID $c$ to the provider peer $p$. 
So, along with the routing query for $c$, the client must also query these peers' provider store for $c$.
We refer to this as a provider advertisement query.

A server peer processes a provider advertisement query non-privately as follows. 
The provider advertisement datastore is keyed by the CID concatenated with the provider peer's ID, since a block with a given CID can be provided by multiple peers. 
Each value in the provider store is the expiry time of the advertisement.
Upon receiving a provider advertisement query, a server peer looks up the provider store for keys prefixed by the given CID. 
It checks if each such provider advertisement has not expired, based on the value, and if so, extracts the provider peer's ID from the key suffix. 
Recall that the server peer also maintains an address book, which maps peer IDs to their  multiaddresses. 
The server peer then consults its address book for the provider peers' IDs and fetches their multiaddresses.
We observe that the client's desired CID continues to leak to the server in the provider advertisement query, since the server looks up its provider advertisement store based on this CID.

\subsection{Private Algorithm Integration}\label{sec.priv-alg-for-providers}
We present our algorithm for private provider advertisement queries in \Cref{alg:non-trivial-private-provider-records}. 
This enables a client to send a PIR query hiding the CID that it is querying for. The server returns a PIR response with the peer IDs of corresponding providers, if they exist in the server's provider store.

\paragraph{Joining Provider Advertisements to Multiaddresses.}
Similar to the non-private peer routing query, the non-private provider advertisement query uses the provider peers' IDs from the provider store as a key to the address book which contains multiaddresses. 
We join the two key-value stores such that for a given key (CID), the value contains the provider peer's multiaddresses (\Cref{alg:prov-routing-get-cid-prov-peerids}--\Cref{alg:prov-routing-join-peerids-with-multiaddresses}).
We recompute the join whenever a provider or peer record is inserted. 
Provider advertisements have a short default expiry period (of 24 hours), and are typically reinserted after expiration.

\paragraph{Binning for PIR.}
    Given that we now have a joined key-value store, we could immediately apply a KeywordPIR scheme, wherein the desired CID acts as the keyword that the server can privately lookup.
    Instead, to avoid costs of KeywordPIR schemes, we transform our joined key-value store to one which enables an IndexPIR scheme, which is more efficient.
    This follows the approach described in \Cref{sec.background-pir}, with optimizations.
    We exploit the fact that the CIDs can be used to directly map each key-value pair from the joined key value store to a hash table.
    CIDs in IPFS are computed using \sha{}, so we can treat them as indices of a hash table.
    All key-value pairs whose CIDs have the same $\ell$-bit prefix are grouped into a single row of the table.
    
    We refer to this process as \emph{binning}, and it enables multiple pairs to be binned into the same row of the table. 
    The bins would fill up non-uniformly as more provider advertisements are added. 
    So the server pads incomplete bins to the size of the largest bin.
    This provides a hash table with $B$ consistently-sized bins over which we can perform PIR.
    While the server could vary the number of bins with the number of provider advertisements that it stores, this would require the client to know the number of bins in advance to formulate its query. 
    Instead, we fix $B$ beforehand, based on the parameterization of the chosen PIR scheme, as we describe later. 
    To retrieve a provider advertisement for a given CID, $c$, the client uses its $\log_2 B$-bit prefix as their desired index in the hash table (\Cref{alg:prov-routing-client-computes-bin-index}).
    It then constructs its PIR query based on this index (\Cref{alg:prov-routing-client-constructs-query}).
    The performance of the PIR protocol is relative to the size of the bins.

\paragraph{Encrypted Entries.}
    Immediately applying PIR to the binned table may enable a client to learn more than it asked for, since any row in the table contains multiple provider advertisements. 
    To prevent this leakage, all advertisements are encrypted under a key derived from the corresponding CID. 
    Effectively, the CID acts as a \emph{pre-shared secret} between a client who desires the provider peers for that CID and a server who stores them.
    We use a Key Derivation Function (KDF) to compute a key $k_j = \mbox{KDF}(c)$ for each CID $c$ in the joined store (\Cref{alg:prov-routing-derive-key}). 
    We then encrypt the multiaddresses for each CID under this key, using a \emph{robust} symmetric encryption scheme (\Cref{alg:prov-routing-encrypt-multiaddresses}).   
    The server then computes the PIR response over the table of these symmetrically encrypted provider advertisements (\Cref{alg:prov-routing-pir-responses}). 
    After a client receives and decrypts the homomorphically encrypted PIR response, they obtain a set of ciphertext provider advertisements, which are symmetrically encrypted under keys derived from CIDs.
    The client can derive the key from their desired CID, $k = \mbox{KDF}(c)$, and use it to decrypt each ciphertext.
    Under a \emph{robust} encryption scheme, only the ciphertext encrypted with the key derived from $c$ will decrypt to a valid plaintext.
    This provides security of a \emph{symmetric PIR} scheme. 
    We note that in a \emph{trivial} symmetric PIR scheme, the server would send back $T$ (Line~\ref{alg:prov-routing-encrypt-multiaddresses}) to the client.
    We discuss optimal choices for PIR schemes for this use case in \Cref{sec:pir-scheme-analysis}.


\begin{algorithm}[t]
    \caption{Private Provider Advertisements using PIR with $B$ bins over a list of provider records $P$}\label{alg:non-trivial-private-provider-records}
    \begin{algorithmic}[1]
    \Procedure{PrivateProvAdQuery}{CID $c$}
        \State $q \leftarrow c[:\log_2 B]$ \Comment{$q\in\{0,1,\cdots,B-1\}$} \label{alg:prov-routing-client-computes-bin-index}
        \State $(\sk, (\pk, \ct)) \leftarrow \pir.\textsc{Query}(q)$ \label{alg:prov-routing-client-constructs-query}
        \State Send $(\pk, \ct)$ to the server.
    \EndProcedure
    \vspace{2mm}
    \Procedure{PrivateProvAdResp}{$(\pk, \ct)$, $P$}
        \State Initialize tables: $J$ with $\|P\|$ rows, $T$ with $B$ bins.
        \For{$i \in \{1,2,\cdots, |P|\}$} 
            \State ($c$, $\{p_1, p_2 \cdots p_j\}) \gets P[i]$  \label{alg:prov-routing-get-cid-prov-peerids} %
            \State $J[i] \gets \emptyset$ \Comment{$J$ contains the join of $P$ with $M$.}
            \For{$i \in \{1,2,\cdots, j\}$} 
                \State $J[i] \gets J[i] \cup M[p_j]$  \label{alg:prov-routing-join-peerids-with-multiaddresses} 
            \EndFor
            \State $b\leftarrow c[:\log_2 B]$ \Comment{Bin index.} \label{alg:prov-routing-compute-bin-index}
            \State $k \gets \mbox{KDF}(c)$ \Comment{Derive a key using the CID.} \label{alg:prov-routing-derive-key}
            \State $T[b] \gets \mbox{SE}.\encrypt(k,J[i])$ \Comment{Encrypt addresses.} \label{alg:prov-routing-encrypt-multiaddresses} %
        \EndFor
        \State $\ans\leftarrow\pir.\textsc{Response}((\pk, \ct), T)$ \label{alg:prov-routing-pir-responses} 
    \EndProcedure
    \end{algorithmic}  
\end{algorithm}

\paragraph{Security.} 
The security of the chosen PIR protocol ensures that server does not learn the queried CIDs.
Meanwhile, the client only receives a set of encrypted provider advertisements.
The \emph{symmetric} property of our protocol requires that the client only learns provider advertisements for the CID that they queried for, and learns nothing about other provider advertisements stored by the server.
In our protocol, this reduces to the security of the robust symmetric encryption scheme and the KDF.
We defer the proof to \Cref{sec:security-proof-ce}.

\section{Private Content Retrieval}\label{sec.privbitswap}
IPFS employs the \emph{Bitswap}~\cite{bitswap} content retrieval protocol. 
As with the provider routing protocol in IPFS, the Bitswap protocol leaks the CID retrieved to the server.
We describe the two steps of the Bitswap protocol (\wanthave{} and \block{}) and introduce private equivalents.
\iffullversion
Our methods ensure that the algorithms scale well with the number of blocks stored.
\fi

\subsection{Bitswap Protocol}
\label{subsec:bitswap-protocol}
\paragraph{The \wanthave{} Step.}
The client concurrently sends a message to multiple candidate peers, asking if they possess a block with a given CID.
If the peer has the requested CID in its table, it responds with a \textit{Have} message.
Otherwise, it may respond with a \textit{DontHave} message if the client asks to receive such a message.
Candidate peers can be obtained through provider advertisement queries. 
Importantly, IPFS clients can use the \wanthave{} step to first query multiple long-lived peers in its RT, thereby assessing whether they have the content block, \emph{before} searching for provider advertisements for that block. 
So, the set of candidate peers can either correspond to a set of peers in the clients' RT that are queried opportunistically, or a set of peers that are output from provider advertisement queries.

\paragraph{The \block{} Step.}
In this step, the client asks a peer who has responded with \textit{Have} in the previous step to transfer the desired block. 
There may not be a clear separation between these two steps, depending on the block size used in the chunking algorithm~\cite{ipfs-chunking}. 
For example, a server peer can immediately send a small block along with the \textit{Have} response, thereby avoiding an extra round of interaction.

\subsection{The Private Bitswap Protocol}
\label{sec:priv-bitswap}
While translating the Bitswap protocol into a private version, we could develop a single-round protocol; however, we observe that the client would incur multiplicative communication overhead when it attempts to use this protocol to opportunistically query multiple peers who may have the desired content block.
So, we choose to maintain the structure of a two-round protocol, with a lightweight \privwanthave{} step and a relatively heavyweight \privblock{} step. 
The \privwanthave{} step can be run opportunistically with multiple candidate peers to determine whether a peer has a content block, while hiding the CID. 
We provide additional information in this step to reduce the overhead of the \privblock{} step. 
While private set intersection (PSI) schemes can be used in the first step, as in \cite{bitswap-privacy2}, they tend to be relatively communication intensive.
In the \privblock{} step, the client privately retrieves the block from the server peer using PIR.
For the remainder of this section, recall the notation introduced in \Cref{sec.notation}.

\paragraph{\privwanthave{}.}
    The goal of this step is for the client to learn if the server has a block with a given CID, and if so, at which index of the server's table does this block lie. 
    A trivial approach would have the server send the entire list $L$ of CIDs to the client, allowing the client to search the list for $c$ and obtain its index $i$ such that $L[i]=c$.
    This would reveal to the client other CIDs that are stored by the server peer.
    This leakage can be resolved by the server computing the hash of every entry in the list using a public hash function. Then, when the client receives the list of all (hashed) CIDs, they need only compute the hash of the CID they desire and determine its index in the list. Under the assumption that the hash function is one-way, the client does not learn which other CIDs are stored by the server.
  
    Alternatively, we can reuse the process for provider advertisements to transmit this information to the client via non-trivial PIR. 
    In particular, for the \privwanthave{} step, the client and server can use the same \Cref{alg:non-trivial-private-provider-records}, 
    using the list of CIDs $L$ in place of the joined provider store $J$.
    Instead of encrypting the provider peers' multiaddresses, we encrypt the index of the CID, and thus Line~\ref{alg:prov-routing-encrypt-multiaddresses} changes to $T[b] \gets \mbox{SE}.\encrypt(k,\mathbf{i})$.

    In the \privwanthave{} step, the server also returns the optimal PIR protocol to use in the \privblock{} step, 
    depending on the number of blocks they hold. This enables an adaptive approach to efficiently retrieve a block.
    Optimal choices of PIR schemes for the \privwanthave{} step are the same as those for the private provider advertisement functionality (\Cref{sec.priv-alg-for-providers}), 
    since we apply the same algorithm.
    These are discussed in \Cref{sec:pir-scheme-analysis}.
    Similarly, the security of the \privwanthave{} is equivalent to the security of the private provider advertisement algorithm. 

\paragraph{\privblock{}.}
Following the \privwanthave{} step, the client knows an index, $j$, of the desired block in the table $D$ held at the server. Thus, it can issue an IndexPIR query to retrieve the desired content block.
The optimal choice of PIR protocol for this step depends on the number of blocks that the peer possesses.
The optimal PIR algorithm is communicated to the client as part of the \privwanthave{} response, enabling the client to correctly construct the query for the \privblock{} step.

Directly applying a PIR scheme to the table $D$ of content blocks is undesirable as it leaves the server peer vulnerable to an index enumeration attack.
A client could query for an arbitrary index and retrieve a block of content without knowing its corresponding CID. 
Thus, we require some cryptographic pre-processing of the data in $D$. 
We reapply a technique that we have previously identified: using the CID as a pre-shared secret between clients and servers. 
Each block in the table can then be encrypted under a key derived from the corresponding CID
so that a client running a PIR query can only decrypt the response if they know the CID. 
We summarize how the server handles a client's PIR query $q$ over $D$ in \Cref{alg:private-block}.

The server only obtains a PIR query from the client, and thus it can only break the security of \privblock{} if it can break the underlying PIR scheme.
On the part of the client, the security argument is equivalent to that for the private provider advertisements algorithm. 
  
\begin{algorithm}[ht]
    \caption{\privblock{} using PIR}\label{alg:private-block}
    \begin{algorithmic}[1]
    \Procedure{PrivateBlockResponse}{$\qu$, $D$}
        \For{$j \in \{1,2,\cdots, m\}$}
            \State $k \gets \mbox{KDF}(L[j])$ \Comment{$L[j]$ is the CID.}\label{alg:private-block-get-SE-key}
            \State $D[j] \gets \mbox{SE}.\encrypt(k,D[j])$ \label{alg:private-block-encrypt}
        \EndFor
        \Return \textsc{PIR.Response}($q$, $D$) 
    \EndProcedure
    \end{algorithmic}   
    \label{alg:private-block-do-pir-response}
\end{algorithm}

\section{Analysis of PIR Protocols}
\label{sec:pir-scheme-analysis}

In this section, we analyze options for PIR protocols to be used in our algorithms from Sections~\ref{sec.privrouting}--\ref{sec.privbitswap}.
Our goal is to select schemes with the best communication--computation cost tradeoff for each of the following use cases:
\begin{enumerate}[nolistsep]
    \item \textit{Few rows with small payloads:} PIR over a table of 256 rows with roughly 1.5KB entries for private routing.
    \item \textit{Bins with varying size:} PIR over bins for provider advertisements and \privwanthave{} steps.
    \item \textit{Many rows with large payloads:} PIR over 256KB payloads, representing content blocks in \privblock{}.
\end{enumerate}

We require that the protocol maintain the same number of rounds as the non-private protocols, to adhere to the same structure.
Hence, the chosen protocol for each case should only require a single round. 
Also, the PIR protocol should be efficient even if only one PIR query is issued by a client to a specific server, which is the case for our applications.
Thus, we cannot amortize costs over many queries.


We examine existing PIR protocols in \Cref{subsec:pir-existing-work}.
While a few of these protocols are reasonable for our \privblock{} step, applying any of these protocols to the prior two use cases would lead to infeasible costs.
To sufficiently address our use cases, we propose two PIR protocols: \basicpaillier{} and \basicrlwe{}, in \Cref{sec.basicpir}.
Both are tailored for the IPFS setting but may be of independent interest in other distributed systems with similar architectures.
We evaluate the costs of our proposed and existing PIR schemes for each use case in \Cref{subsec:eval}.

\subsection{Existing PIR Protocols}
\label{subsec:pir-existing-work}
Modern PIR protocols typically consist of offline and online phases.
In the offline phase, content that is independent of the query is exchanged (e.g., cryptographic keys and database-dependent hints) to accelerate the online phase.
In the online phase, the server computes a response to a query.
To adhere to having only one round of communication, we require that any offline phase which requires interaction with the client is conducted concurrently with the online phase.
\iffullversion
In \Cref{app:extended-pir-scheme-analysis}, we elaborate on how this limitation precludes many PIR protocols in the literature for our application.
\fi

We observe that PIR protocols which use client-specific cryptographic keys are the only suitable candidates (we elaborate on this in
\iffullversion
\Cref{app:extended-pir-scheme-analysis}).
\fi
This leaves SealPIR~\cite{sealpir}, FastPIR~\cite{fastpir}, OnionPIR~\cite{onionpir}, Spiral~\cite{spiral}, HintlessPIR~\cite{liHintlessSingleServerPrivate2024}, and YPIR~\cite{menonYPIRHighThroughputSingleServer2024a}. 
Table~\ref{tab:pir-results} summarizes the high level costs of each of these protocols.
Spiral provides small queries and responses but requires the client to send a large, 13MB cryptographic key. 
In contrast, SealPIR, FastPIR, and OnionPIR involve smaller 
cryptographic keys but require larger queries and/or responses. 
HintlessPIR requires no client-specific keys but has very large response sizes. 
YPIR requires very small keys but is optimized for payloads smaller than our use cases.
Importantly, we deemed all existing schemes to be insufficient for our first and second use cases, due to incompatible payload sizes.
Our \basicpaillier{} scheme supports smaller payloads, which is relevant for peer routing.
Variants of \basicrlwe{} produce the smallest client-specific cryptographic keys, in comparison to the aforementioned schemes, which is useful for the private provider advertisements and \privwanthave{} steps.
We do consider existing PIR schemes which support large, 256 KB payloads, for the \privblock{} step and evaluate them in \Cref{subsec:eval}.

\begin{table}[t]
    \centering
    \caption{
        Lower bounds on the size of the key material, query, and response size of various protocols.
        The second column denotes the size of client-specific cryptographic keys used in each approach.
        *In HintlessPIR and YPIR, the lower bound on communication depends on the size of the database. In these cases, we assume the database is 1 GB.
        }
\resizebox{\columnwidth}{!}{
    \begin{tabular}{c|c|c|c|c}
    \toprule
        \multirow{3}{*}{Protocol}
         & Key & Query & \multicolumn{2}{c}{Response} \\
         & Material & (Encrypted) & (Plaintext) & (Encrypted) \\
    \midrule
        SealPIR \cite{sealpir}      & 1.6 MB & 90 KB & 10 KB  & 181 KB \\
        FastPIR  \cite{fastpir}     & 0.67 MB &  64 KB & 10 KB  & 65 KB \\
        OnionPIR \cite{onionpir}    & 5.4 MB & 64 KB & 30 KB  & 128 KB \\
        Spiral \cite{spiral}        & 13 MB & 28 KB & 7.5 KB & 20 KB \\
        HintlessPIR$^{*}$~\cite{liHintlessSingleServerPrivate2024} & - & 453 KB & 32 KB & 3080 KB\\
        YPIR$^{*}$~\cite{menonYPIRHighThroughputSingleServer2024a} & 462 KB & 384 KB & 1 B & 12 KB \\
        \hline 
        \basicpaillier{} & 1.14 KB & 0.38 KB & 0.38 KB  & 0.76 KB \\
        \basicrlwe{} & 750 KB & 64 KB & 7.5 KB  & 65 KB \\
        \rlwethree{} & 192 KB &   64 KB & 7.5 KB  & 65 KB \\
        \rlwetwo{} & 128 KB  &    64 KB & 7.5 KB  & 65 KB \\
    \bottomrule
    \end{tabular}
}
    \label{tab:pir-results}
\end{table}

\begin{algorithm}[!t]
    \caption{
        \basicpaillier{} and \basicrlwe{}.
        The composite modulus of Paillier is denoted as $M$.
        The plaintext and ciphertext space of RLWE are denoted as $R_p$ and $\mathcal{C}$, respectively.
        $\db$ represents a database with $n$ rows and $\ell$ columns.
        We assume the size of each cell in the database is the size of one plaintext in the corresponding scheme, i.e., in \basicpaillier, $\db\in\ZZ_M^{n\times \ell}$ and in \basicrlwe, $\db\in R_p^{n\times \ell}$.
    }\label{alg:basic-pir}
    \begin{algorithmic}[1]

    \Procedure{$\basicpaillier.\query$}{$i$}\Comment{$i\in\{1,\cdots,n\}$}
        \State Sample Paillier secret key $\sk$
        \State $\ct \sample \ZZ_{M^2}^{n}$\label{alg:sample-paillier-cts}
        \For{$j\in\{1,\dots,n\}$}
            \State $b_j \leftarrow \mathbb{I}[i=j]$
            \label{alg:basic-paillier-start-beck}
            \State $r_j \leftarrow\pailliercrypto.\decrypt(\sk, \ct[j])$ \label{alg:basic-paillier-beck-decryption}\Comment{$r_j\in\ZZ_M$}
            \State $\pt[j] = b_j-r_j \mod M$\label{alg:basic-paillier-end-beck}\Comment{$\pt\in\ZZ_M^{n}$}
        \EndFor
        \Return $(\sk, (\ct, \pt))$
    \EndProcedure
    \vspace{2mm}
    \Procedure{$\basicpaillier.\response$}{$(\ct,\pt)$, $\db$} \label{alg:paillier-response}
        \For {$j\in \{1,\dots,n\}$}
            \State $\ct'[j] \leftarrow \pailliercrypto.\add(\ct[j], \pt[j])$
            \label{alg:basic-pir-paillier-beck-expand}
        \EndFor
        \State $\ans \leftarrow [0]*\ell$
        \For{$k\in\{1,\dots,\ell\}$}
            \State $\ans[k] \leftarrow \pailliercrypto.\scalmult(\ct'[1], \db[1][k])$ \label{alg:paillier-scalarmult-1}
            \For{$j\in \{2,\dots,n\}$}
                \State $t \leftarrow \pailliercrypto.\scalmult(\ct'[j], \db[j][k])$ \label{alg:paillier-scalarmult-2}
                \State $\ans[k] \leftarrow \pailliercrypto.\add(\ans[k], t) $
            \EndFor
        \EndFor
        \Return $\ans$
    \EndProcedure
    \vspace{2mm}
    \Procedure{$\basicrlwe.\query$}{$i$}\Comment{$i\in\{1,\cdots,n\}$}
        \State $\sk, \pk \leftarrow \algostyle{GenerateKeys}()$
        \label{alg:auto-keygen}
        \For{$j \in \{1, 2 .. \ceil{n/N}\}$}
            \If {$j*N \leq i < (j+1)*N$}
                \State $\ct[j] \leftarrow \rlwe.\encrypt(\sk, X^{i\mod N})$
            \Else 
                \State $\ct[j] \leftarrow \rlwe.\encrypt(\sk, 0)$
            \EndIf
        \EndFor
        \Return $(\sk, (\pk, \ct))$
    \EndProcedure
    \vspace{2mm}
    \Procedure{$\basicrlwe.\response$}{($\pk$, $\ct$), $\db$}
        \State $C\leftarrow []$
        \For{$j \in \{1, 2 .. \ceil{n/N}\}$}
            \State $t \leftarrow\oblivexpand(\pk, \ct[j])$ 
            \Comment{$t\in\mathcal{C}^N$}
            \label{alg:obliv-expand}
            \State Append ciphertexts in $t$ to $C$
        \EndFor
        \State $\ans \leftarrow [0]*\ell$
        \For{$k\in\{1,\cdots,\ell\}$}
            \State $\ans[j] \leftarrow \rlwe.\scalmult(C[1], \db[j][k])$
            \For{$j\in \{2,\dots,n\}$}
                \State $t\leftarrow\rlwe.\scalmult(C[j], \db[j][k])$
                \State $\ans[k] \leftarrow \rlwe.\add(\ans[k], t)$
            \EndFor
        \EndFor
        \Return $\ans$
    \EndProcedure
    \end{algorithmic}    
\end{algorithm}

\subsection{Tailored Protocols: \basicpaillier{}, \basicrlwe{}}\label{sec.basicpir}
Our protocols are based on Paillier or RLWE encryption, and use existing techniques from cryptographic folklore.
Specifically, both protocols compute an inner product between an encrypted indicator vector, supplied by the client, and a table.
The protocols are straightforward in nature, which both allows for simpler integration into large systems and demonstrates the lack of attention paid to such use cases in the PIR literature.
They are described in \Cref{alg:basic-pir}.

\paragraph{\basicpaillier{}.}
This protocol is based on the Paillier cryptosystem and uses homomorphic additions ($\pailliercrypto.\add$) and scalar multiplications ($\pailliercrypto.\scalmult$). 
The client need only send a small public key (1\,KB) to the server along with the PIR query; this key  includes the Paillier composite modulus and the chosen generator.
For the PIR query, the client encrypts an indicator vector corresponding to the desired row, which produces one ciphertext for each row in the table.
To reduce the size of the query, we employ a technique from Beck~\cite{beckPaillier}.
Instead of directly sending the indicator ciphertext vector, the client samples a random ciphertext vector (\Cref{alg:sample-paillier-cts}). 
It then decrypts this ciphertext vector to obtain a randomized plaintext vector (\Cref{alg:basic-paillier-beck-decryption}). 
It uses this plaintext vector to mask each element of the indicator vector (\Cref{alg:basic-paillier-end-beck}) and sends the masked plaintext vector to the server.
Instead of having the client send the randomized ciphertext vector (\Cref{alg:sample-paillier-cts}), in practice the client only sends a seed used to generate this vector.
The communication cost for one query is then the size of the seed in addition to the size of a vector of \emph{plaintexts}. 
This technique reduces the communication cost of a query by approximately half.

The server provides a PIR response using the client's query (Line~\ref{alg:paillier-response}).
In doing so, it runs a scalar multiplication for each element in the database (Lines~\ref{alg:paillier-scalarmult-1},~\ref{alg:paillier-scalarmult-2}).
The computation overhead is thus proportional to the size of the database, making it impractical for large databases.
Nonetheless, this protocol is useful for databases with few, small rows. 
Proofs of correctness and security for \basicpaillier{} are provided in \Cref{sec:prove-paillier}.
In our evaluation, we choose a 3072-bit composite modulus for 128-bit security~\cite{paillier-parameters, barkerRecommendationKeyManagement2020}.

\paragraph{\basicrlwe{}.}
Our second construction is based on the RLWE assumption~\cite{lyubashevskyIdealLatticesLearning2010} and related additive homomorphic schemes~\cite{brakerski_fully_2012,fan2012somewhat,brakerski_leveled_2012}. 
Plaintexts are polynomials $p(x)\in R_p = \ZZ_p/(X^N+1)$.
We make use of homomorphic addition ($\rlwe.\add$) and scalar multiplication ($\rlwe.\scalmult$).
Our construction combines common techniques in the literature for RLWE-based PIR protocols~\cite{sealpir,de2024whispir, fastpir}, and tunes them to our setting. 
Compared to \basicpaillier{},  \basicrlwe{} has lower computation overhead but higher communication costs.
Suppose the client wants to request row $i$. It generates the plaintext polynomial $p(X)=X^i$, and sends the encryption of this polynomial to the server as the PIR query. 
The server then \emph{derives} an encrypted indicator vector from the query using an oblivious expansion technique~\cite{sealpir, de2024whispir, chenOnionRingORAM2019}.

The functionality of oblivious expansion is as follows: on an input ciphertext $c$ which encrypts $p(X)=a_0+a_1X+\cdots+a_{N-1}X^{N-1}$, it outputs $N$ ciphertexts, $\{c_i\}_{i=0,\cdots,N-1}$ such that $c_i$ encrypts the scalar $a_i$.
An automorphism key, for a given $k\in\mathbb{Z}_{2N}^{*}$ allows the server to compute the encryption of $p(X^k)$ from the encryption of $p(X)$.
Angel et al.~\cite{sealpir} first proposed an oblivious expansion scheme that required $\log_2 N$ automorphism keys to be sent from the client to the server~\cite{sealpir, chenOnionRingORAM2019}.
More recently, de Castro et al.~\cite{de2024whispir} propose a change in oblivious expansion 
which requires only a constant number of automorphism keys.
While this reduces the communication cost, the server needs to perform more automorphisms, thereby increasing the computation cost.
We refer to the server-side oblivious expansion technique by $\oblivexpand$ in \Cref{alg:obliv-expand} and the client-side procedure to generate the necessary secret and auxiliary keys by $\algostyle{GenerateKeys}$ in \Cref{alg:auto-keygen}.

In addition to the $\basicrlwe{}$, we develop two variants of oblivious expansion based on de Castro et al.'s construction~\cite{de2024whispir}. 
\rlwetwo{} and \rlwethree{} involve two and three automorphism keys, respectively.
Our \basicrlwe{} variant sends more keys; we detail these variants and their application to the first and second use cases in \Cref{sec:oblivious-expand-variant}. 
In all variants of \basicrlwe{}, we use a polynomial modulus degree of $N=4096$ and set the plaintext modulus to $40961$.
The ciphertext modulus is the composite of a 54-bit and 55-bit prime.
These parameters ensure correctness and provide 128-bit security.
The security of our protocol follows directly from that of prior work~\cite{sealpir, de2024whispir}.

\def\bitswapResultsDir{data/results-chippie}
\def\routingResultsDir{data/results-moone}

\begin{figure*}[ht]
    \centering
    \begin{tikzpicture}
        \begin{groupplot}[
          group style={
            group size=3 by 2,
            group name=plots,
            horizontal sep=1.8cm,
            vertical sep=0.6cm
          },
          width=0.65\columnwidth,height=0.45\columnwidth
        ]
        \nextgroupplot[
            xmode=log,
            ymode=log,
            ymin=10, 
            ymax=1000,
            legend pos=outer north east,
            ylabel={Communication (KB)},
            ylabel near ticks,
            xtick={16,32,64,128,256},
            xticklabels={16,32,64,128,256},
            log ticks with fixed point,
        ]

        \addplot [teal, thick] table [x=NumRows, y expr =\thisrow{TotalLenMean(Bytes)} / 1024, col sep=comma] {\routingResultsDir/peerRouting-RLWE_All_Keys.csv};

        \addplot [green, thick] table [x=NumRows, y expr =\thisrow{TotalLenMean(Bytes)} / 1024, col sep=comma] {\routingResultsDir/peerRouting-RLWE_Whispir_3_Keys.csv};

        \addplot [orange, thick] table [x=NumRows, y expr =\thisrow{TotalLenMean(Bytes)} / 1024, col sep=comma] {\routingResultsDir/peerRouting-RLWE_Whispir_2_Keys.csv};
        
        \addplot [blue, thick] table [x=NumRows, y expr= \thisrow{NumRows} * 20 * 81 / 1024, col sep=comma] {\routingResultsDir/peerRouting-RLWE_All_Keys.csv};

        \addplot [black, thick] table [x=NumRows, y expr =\thisrow{TotalLenMean(Bytes)} / 1024, col sep=comma] {\routingResultsDir/peerRouting-Basic_Paillier.csv};
    
        \nextgroupplot[
            legend pos=north east,
            ymax=1.6,
            ylabel={Communication (MB)},
            ylabel near ticks,
            xtick = {8,40,80,120,160,200}, 
            xticklabels={8k,40k,80k,120k,160k,200k},
          ]
        \addplot [green, thick] table [x expr= (\coordindex + 1)*8, y expr=\thisrow{TotalLenMean(Bytes)} / (1024 * 1024), col sep=comma] {\routingResultsDir/providerRouting-RLWE_Whispir_3_Keys.csv};

        \addplot [orange, thick] table [x expr= (\coordindex + 1)*8, y expr=\thisrow{TotalLenMean(Bytes)} / (1024 * 1024), col sep=comma] {\routingResultsDir/providerRouting-RLWE_Whispir_2_Keys.csv};

        \addplot [teal, thick] table [x expr= (\coordindex + 1)*8, y expr=\thisrow{TotalLenMean(Bytes)} / (1024 * 1024), col sep=comma] {\routingResultsDir/providerRouting-RLWE_All_Keys.csv};

        \addplot [black, thick] table [x expr= (\coordindex + 1)*8, y expr=\thisrow{TotalLenMean(Bytes)} / (1024 * 1024), col sep=comma] {\routingResultsDir/providerRouting-Basic_Paillier.csv};

        \addplot [yellow, thick] table [x expr= (\coordindex + 1)*8, y expr= \thisrow{NumRows} * 100 / (1024 * 1024), col sep=comma] {\routingResultsDir/providerRouting-RLWE_All_Keys.csv};

        \nextgroupplot[
            xmode=log,
            legend pos= north west,
            ylabel={Communication (MB)},
            ylabel near ticks,
            ytick = {0, 5*1024, 10*1024, 15*1024, 20*1024
            }, 
            yticklabels={0, 5, 10, 15, 20,
            },
            minor y tick num=5,
            ytick scale label code/.code={},
            ymax = 16*1024,
        ]
        \addplot [cyan, thick] table [x=num_rows, y=communication_KB, col sep=comma] {\bitswapResultsDir/sealpir.csv};
        \addplot [violet, thick] table [x=num_rows, y=communication_KB, col sep=comma] {\bitswapResultsDir/fastpir.csv};
        \addplot[pink, thick] table [x=num_rows, y=communication_KB, col sep=comma] {\bitswapResultsDir/onionpir.csv};
        \addplot[teal, thick] table [x=num_rows, y=communication_KB, col sep=comma] {\bitswapResultsDir/RLWE_All_Keys.csv};
        \addplot[red, thick] table [x=num_rows, y=communication_KB, col sep=comma] {\bitswapResultsDir/spiral-stream-pack.csv};

        \nextgroupplot[
            xmode=log,
            log basis x={2},
            ymode=log,
            legend pos=outer north east,
            xlabel={Number of rows in normalized RT},
            xtick={16,32,64,128,256},
            xticklabels={16,32,64,128,256},
            log ticks with fixed point,
            ylabel={Runtime (s)},
            ylabel near ticks,
            legend to name=biglegend,
            legend style={font=\small,legend columns=5,/tikz/every even column/.append style={column sep=1.0cm}}, 
        ]

        \addplot [black, thick, error bars/.cd, y dir=both,y explicit] table [x=NumRows, y expr=\thisrow{ServerRuntimeMean(ms)}/1000, y error expr= \thisrow{ServerRuntimeStddev(ms)}/1000, col sep=comma] {\routingResultsDir/peerRouting-Basic_Paillier.csv};
        \addlegendentry{{\color{black}\basicpaillier{}}}

        \addplot [orange, thick, error bars/.cd, y dir=both,y explicit] table [x=NumRows, y expr=\thisrow{ServerRuntimeMean(ms)}/1000, y error expr= \thisrow{ServerRuntimeStddev(ms)}/1000, col sep=comma]  {\routingResultsDir/peerRouting-RLWE_Whispir_2_Keys.csv};
        \addlegendentry{{\color{black}\rlwetwo{}}}

        \addplot [green, thick, error bars/.cd, y dir=both,y explicit] table [x=NumRows, y expr=\thisrow{ServerRuntimeMean(ms)}/1000, y error expr= \thisrow{ServerRuntimeStddev(ms)}/1000, col sep=comma] {\routingResultsDir/peerRouting-RLWE_Whispir_3_Keys.csv};
        \addlegendentry{{\color{black}\rlwethree{}}}

        \addplot [teal, thick, error bars/.cd, y dir=both,y explicit] table [x=NumRows, y expr=\thisrow{ServerRuntimeMean(ms)}/1000, y error expr= \thisrow{ServerRuntimeStddev(ms)}/1000, col sep=comma] {\routingResultsDir/peerRouting-RLWE_All_Keys.csv};
        \addlegendentry{{\color{black}\basicrlwe{}}}

        \addlegendimage{cyan, thick}
        \addlegendentry{{\color{black}SealPIR}}
        \addlegendimage{violet, thick}
        \addlegendentry{{\color{black}FastPIR}}
        \addlegendimage{pink, thick}
        \addlegendentry{{\color{black}OnionPIR}}
        \addlegendimage{red, thick}
        \addlegendentry{{\color{black}Spiral}}
        \addlegendimage{blue, thick}
        \addlegendentry{{\color{black}\trivialpir{}}}
        \addlegendimage{yellow, thick}
        \addlegendentry{{\color{black}\trivialspir{}}}

        \nextgroupplot[
            ymin=0.5, ymax=1.8,
            minor ytick={0.5, 0.6, 0.7, 0.8, 0.9, 1.0, 1.1, 1.2, 1.3, 1.4, 1.5, 1.6, 1.7, 1.8},
            legend pos=north east,
            xlabel={Number of provider advertisements}, 
            ylabel={Runtime (s)},
            xtick = {8,40,80,120,160,200}, 
            xticklabels={8k,40k,80k,120k,160k,200k},
            ylabel near ticks,
          ]

        \addplot [orange, thick, error bars/.cd, y dir=both,y explicit,] table [x expr= (\coordindex+1)*8, y expr=\thisrow{ServerRuntimeMean(ms)}/1000, y error expr= \thisrow{ServerRuntimeStddev(ms)}/1000, col sep=comma] {\routingResultsDir/providerRouting-RLWE_Whispir_2_Keys.csv};
        %
        
        \addplot [green, thick, error bars/.cd, y dir=both,y explicit,] table [x expr= (\coordindex+1)*8, y expr=\thisrow{ServerRuntimeMean(ms)}/1000, y error expr= \thisrow{ServerRuntimeStddev(ms)}/1000, col sep=comma] {\routingResultsDir/providerRouting-RLWE_Whispir_3_Keys.csv};

        \addplot [teal, thick, error bars/.cd, y dir=both,y explicit,] table [x expr= (\coordindex+1)*8, y expr=\thisrow{ServerRuntimeMean(ms)}/1000, y error expr= \thisrow{ServerRuntimeStddev(ms)}/1000, col sep=comma] {\routingResultsDir/providerRouting-RLWE_All_Keys.csv};

        \nextgroupplot[
            ymin=0.1,
            ymax=5000,
            xmode=log,
            ymode=log,
            log ticks with fixed point,
            xticklabels={0, $10^3$,$10^4$,$10^5$,$10^6$},
            legend pos= north west,
            xlabel={Number of content blocks},
            ylabel={Runtime (s)},
            ylabel near ticks,
            minor ytick={0.2,0.3,0.4,0.5,0.6,0.7,0.8,0.9,2,3,4,5,6,7,8,9,20,30,40,50,60,70,80,90,200,300,400,500,600,700,800,900},
            extra y ticks={10,1000},
        ]

        \addplot [cyan, thick, error bars/.cd, y dir=both,y explicit,] table [x expr=\thisrow{num_rows}, y expr=\thisrow{runtime_s_mean}, y error expr= \thisrow{runtime_std}, col sep=comma] {\bitswapResultsDir/sealpir.csv};
          
        \addplot [violet, thick, error bars/.cd, y dir=both,y explicit,] table [x expr=\thisrow{num_rows}, y expr=\thisrow{runtime_s_mean}, y error expr= \thisrow{runtime_std}, col sep=comma] {\bitswapResultsDir/fastpir.csv};
          
        \addplot[pink, thick] table [x=num_rows, y=runtime_s_mean, col sep=comma] {\bitswapResultsDir/onionpir.csv};

        \addplot [teal, thick, error bars/.cd, y dir=both,y explicit,] table [x expr=\thisrow{num_rows}, y expr=\thisrow{runtime_s_mean}, y error expr= \thisrow{runtime_std}, col sep=comma] {\bitswapResultsDir/RLWE_All_Keys.csv};  
        
        \addplot [red, thick, error bars/.cd, y dir=both,y explicit,] table [x expr=\thisrow{num_rows}, y expr=\thisrow{runtime_s_mean}, y error expr= \thisrow{runtime_std}, col sep=comma]{\bitswapResultsDir/spiral-stream-pack.csv};  
        
        \end{groupplot}
        \ifacmversion
\else
        \node[below right,inner sep=0pt] at ([yshift=-4mm,xshift=-3cm]plots c1r2.outer south) {\ref{biglegend}};
        \fi
    \end{tikzpicture}%
    {\phantomsubcaption\label{fig:peer-routing-comm}}%
    {\phantomsubcaption\label{fig:pir-with-binning-communication-costs}}%
    {\phantomsubcaption\label{fig:bitswap-comm}}%
    {\phantomsubcaption\label{fig:peer-routing-runtime}}%
    {\phantomsubcaption\label{fig:pir-with-binning-runtimes}}%
    {\phantomsubcaption\label{fig:bitswap-runtimes}}%
    \caption{
        Communication costs (top graphs) and computation costs (bottom graphs) for private peer routing (left graphs), private provider advertisements (middle graphs), and private content retrieval (right graphs). 
        }
    \label{fig:the-one-to-rule-them-all}
\end{figure*}

\subsection{Evaluation}
\label{subsec:eval}

\paragraph{Experimental Setup.}
We evaluate the total communication cost and the server-side runtime for our three use cases in Figure~\ref{fig:the-one-to-rule-them-all}.
We perform our peer and provider routing experiments on an AMD Ryzen 9 7900X with 12 cores and our content retrieval experiment on an AMD EPYC 7302 with 16 cores. 
We report averages and standard deviations over $N=10$ runs. 

The \textit{Runtime} metric denotes the server-side runtime to process a single private query in each use case.
We parallelized our server-side runtimes for peer and provider routing to minimize latency.
Note that client-side runtime is typically small in comparison to the server-side runtime. 
The \textit{Communication} metric denotes the total round-trip network cost of one private query.
For peer routing, we evaluate these overheads only for a single ``hop'' in the iterative routing process; in a complete lookup, the client would query server peers in multiple hops to reach a target peer. 

We hypothesize the following and examine the validity of these hypotheses in all three of our use cases:

\begin{itemize}[noitemsep,leftmargin=17pt]
    \item[H1]\label{h1} In order of increasing communication cost, we expect \basicpaillier{} followed by \basicrlwe{}. 
    \item[H2]\label{h2} In order of increasing communication cost, we expect our RLWE variants to rank as: \rlwetwo{}, \rlwethree{} and \basicrlwe{}, as the number of automorphism keys sent increases.
    \item[H3]\label{h3} In order of increasing server runtimes, we expect \basicrlwe{} followed by \basicpaillier{}. 
    \item[H4]\label{h4} In order of increasing server runtimes, we expect our RLWE variants to rank as: \basicrlwe{}, \rlwethree{} and \rlwetwo{}.  We expect \rlwetwo{} to be slower than \rlwethree{} since the server performs additional automorphisms, as it has one less automorphism key. 
\end{itemize}


\paragraph{Private Peer Routing.}
We recall this case from \Cref{sec:private-routing-pir-integration} where we query a database of up to 256 rows with 1.5 KB entries.
First, all hypotheses hold across the communication and computation cost plots for this use case. 
Second, each of our new PIR algorithms has a constant communication overhead, even as the number of rows increases. 
This is because the client issues its query assuming the worst case scenario: that the server's RT has 256 buckets.
Third, we also plot the communication overhead when the server sends the entire normalized RT, in the \trivialpir{} plot.
\basicpaillier{} outperforms \trivialpir{} as long as the normalized RT has more than approximately 64 rows. 
However, \trivialpir{} may increase the risk of eclipse attacks, and is thus unsuitable for our threat model.

Finally, in terms of computation overheads, 
 our \basicrlwe{} algorithm is bottlenecked by the server response computation, and thus its runtime increases slightly with the number of rows. 
Meanwhile, our \rlwetwo{} and \rlwethree{} variants are bottlenecked by the oblivious expansion step. 
Regarding hypothesis H3, the gap between \rlwethree{} and \basicrlwe{} narrows as the number of rows increases. 
Based on these observations, we recommend using \rlwethree{} for this use case, as it provides the best communication--computation tradeoff for our goal.

\paragraph{Private Provider Advertisement.}
We evaluate PIR protocols for private provider advertisements (\Cref{sec.prov-ads}); our results can be extended to the \privwanthave{} step (\Cref{sec.privbitswap}).
Recall that in both cases, we group multiple records into bins and perform PIR over the bins.

Only the second and last hypotheses are met for this use-case.
H1 does not hold since the communication cost of \basicpaillier{} scales proportional to the number of bins, which has been parameterized for our RLWE schemes, 
\iffullversion
as we detail in \Cref{app:implementation-details}.
\fi
Regarding H4, we find negligible difference between the runtimes of \rlwethree{} and \basicrlwe{}, because in our parallelized implementation, the number of sequential automorphisms required in \rlwethree{} is only slightly more than \basicrlwe{}.
Third, the communication costs and runtimes increase at approximately $40k$ and $140k$ CIDs for \basicrlwe{}.
This is because as bins fill with more provider advertisements, the server needs to compute multiple ciphertexts to encode a bin in a PIR response.
Fourth, we consider a trivial symmetric PIR approach, namely \textsc{TrivialSPIR} (defined in Section~\ref{sec.priv-alg-for-providers}); this protocol quickly exceeds the communication costs of our protocols.
In conclusion, comparing the variants of \basicrlwe{}, we again recommend using \rlwethree{}.

\paragraph{Private Content Retreival.}
Finally, we examine the case for large content blocks sent during the \privblock{} step. 
In the \privwanthave{} step, the server will output the protocol to use for the \privblock{} step. 
So, we recommend a variety of protocols for this use case, depending on if it is a lightweight server with few content blocks or a heavyweight server with more content blocks.
Our \basicpaillier{} and \basicrlwe{} protocols have been designed for the two aforementioned use cases, and provide a poor communication--computation tradeoff for this case.
Thus, we exclude these options and compare the remaining options from \Cref{tab:pir-results} for this step.

We observe the following from \Cref{fig:the-one-to-rule-them-all}. 
Spiral enables reasonably low runtimes (between 0.4--20s)
for the ranges of database sizes that we consider, while incurring a high communication overhead. 
SealPIR has almost half as much communication overhead as Spiral, and is at least an order of magnitude slower.
At the other extreme of the communication--computation tradeoff, FastPIR has the lowest communication overhead until 100,000 blocks, and this overhead then scales linearly for larger databases.
FastPIR is noticeably slower than other protocols until 10,000 blocks. 

We thus recommend Spiral for the \privblock{} step in Bitswap, due to its significantly lower computation overhead. 
Alternately, to reduce communication costs, lightweight peers that hold less than 10,000 blocks may opt for FastPIR, whereas heavyweight peers can opt for SealPIR.
\section{Discussion \& Takeaways}
\label{sec.discussion}
We summarize our recommendations for each use case in \Cref{tab:recommendations}. 
We provide a single recommendation for the peer routing and provider advertisement use case but offer several recommendations for content retrieval.
While we do incur an overhead, our private algorithms for each use case are agnostic of the ingredient protocol.
Hence, improved PIR protocols in future work can be plugged into our algorithms.

\begin{table}[ht]
    \centering
    \resizebox{\columnwidth}{!}{%
    \begin{tabular}{lcccc}
        \toprule
        \textbf{Use Case} & \textbf{Recommendation} & \textbf{Communication} & \textbf{Runtime} & \textbf{Comment}\\
        \midrule
        Peer Routing & \rlwethree{} & 295 KB & 60 ms \\
        \midrule
        Provider & \multirow{2}{*}{\rlwethree{}} & \multirow{2}{*}{0.43 MB} & \multirow{2}{*}{$<$1 s} \\
        Advertisements & & & \\
        \midrule
        & Spiral & 10--15 MB & 0.5--15s \\
        Content Retrieval & FastPIR & 1.7--2.1 MB & 25--38s  & $10^3$--$10^4$ blocks\\
        & SealPIR & 6.1--6.4 MB & 200--2250s & $10^5$--$10^6$ blocks \\
        \bottomrule
    \end{tabular}
    }
    \caption{Recommended PIR Protocols}
    \label{tab:recommendations}
\end{table}
Based on our design of PIR protocols for the distributed setting, we present the following takeaways.
\begin{itemize}[leftmargin=0.25cm]
        \item Symmetric PIR: Robust symmetric encryption can be used to instantiate a symmetric PIR scheme, by interpreting the index of the record to fetch, which is the CID in our case, as a pre-shared secret between the client and server. 
        \item IndexPIR: Given that CIDs are themselves hashes, we construct an IndexPIR scheme, where the client can compute the bin index by taking a prefix of the CID. Consequently, an extra round to privately fetch the index is unnecessary.
        \item RLWE: Reducing the number of automorphism keys sent by the client, by extending existing oblivious expansion techniques, allows us to gain runtimes comparable to RLWEPIR with much lower communication overhead. 
    \end{itemize}

\paragraph{Impact of Churn.}
We consider two types of churn in the network: node churn, i.e., when a peer enters or leaves the network, and content churn, i.e., when a content block is added or removed by its provider. 
In \Cref{app:churn}, we outline how these types of churn impact each of our algorithms. 
We show that each algorithm succeeds under the applicable churns, and we analyze whether it would incur a higher latency under churn.
At a high level, the latencies of each of our private algorithms are negligible in comparison to the rate of the applicable churns.
Thus, the likelihood that our algorithm would fail in reaching a target peer or in retrieving a target CID, due to churn, would be similar to that in the current IPFS deployment.

\paragraph{Non-collusion Assumptions in DHTs.}
Distributed settings are often seen as ideal applications of non-collusion assumptions, which can enable efficient, information-theoretic algorithms.
For example, Mazmudar et al.~\cite{dhtpir} leverage an IT-PIR protocol for content retrieval in DHTPIR.
They observed that existing work in robust DHTs~\cite{rcp} relied on the ability to partition peers into quorums: groups of peers such that each group has a bounded proportion of malicious peers. 
They exploit quorums to satisfy the non-collusion assumption for IT-PIR.

    However, we found that instantiating these assumptions in IPFS was challenging since robust DHTs relied on theoretical quorum formation protocols. 
    Mazmudar et al.~\cite{dhtpir} suggest forming quorums using the Commensal Cuckoo Rule (CCR)~\cite{ccr}. 
    CCR requires peers to repeatedly leave and join the network until they join a sufficiently big quorum, and thereby, drastically increases the node churn. 
    Even if one were able to efficiently instantiate quorums, applying IT-PIR still appears infeasible, since it requires all peers within a quorum to share files.
    In IPFS, peers decide which files to store, whereas IT-PIR requires them to increase storage by a factor of the size of a quorum.
    This introduces a prohibitive overhead of copying files every time a peer changes quorums.
    Thus, despite appearing to be a natural setting for threshold cryptography, we conclude that the assumptions required in prior work do not hold for IPFS. 


\bibliographystyle{plain}
\bibliography{references}

\begin{thebibliography}{10}

\bibitem{pantheon}
Ishtiyaque Ahmad, Divyakant Agrawal, Amr~El Abbadi, and Trinabh Gupta.
\newblock Pantheon: Private retrieval from public key-value store.
\newblock {\em Proc. VLDB Endow.}, 16(4):643--656, 2022.

\bibitem{fastpir}
Ishtiyaque Ahmad, Yuntian Yang, Divyakant Agrawal, Amr~El Abbadi, and Trinabh
  Gupta.
\newblock Addra: Metadata-private voice communication over fully untrusted
  infrastructure.
\newblock In {\em 15th {USENIX} Symposium on Operating Systems Design and
  Implementation ({OSDI} 21)}, pages 313--329. {USENIX} Association, July 2021.

\bibitem{Ambainis2000}
Andris Ambainis.
\newblock Upper bound on the communication complexity of private information
  retrieval.
\newblock In Pierpaolo Degano, Roberto Gorrieri, and Alberto
  Marchetti-Spaccamela, editors, {\em Automata, Languages and Programming},
  ICALP 1997, pages 401--407, Berlin, Heidelberg, 1997. Springer.

\bibitem{sealpir}
Sebastian Angel, Hao Chen, Kim Laine, and Srinath Setty.
\newblock {PIR} with compressed queries and amortized query processing.
\newblock In {\em 2018 IEEE Symposium on Security and Privacy (SP)}, pages
  962--979, 2018.

\bibitem{qp}
Michael Backes, Ian Goldberg, Aniket Kate, and Tomas Toft.
\newblock Adding query privacy to robust {DHTs}.
\newblock In {\em Proceedings of the 7th ACM Symposium on Information, Computer
  and Communications Security}, ASIACCS '12, page 30–31, New York, NY, USA,
  2012. Association for Computing Machinery.

\bibitem{barkerRecommendationKeyManagement2020}
Elaine Barker.
\newblock Recommendation for key management: Part 1 -- general.
\newblock Technical Report NIST SP 800-57pt1r5, {National Institute of
  Standards and Technology}, Gaithersburg, MD, 2020.

\bibitem{beckPaillier}
Martin Beck.
\newblock Randomized decryption ({RD}) mode of operation for homomorphic
  cryptography - increasing encryption, communication and storage efficiency.
\newblock In {\em 2015 IEEE Conference on Computer Communications Workshops
  (INFOCOM WKSHPS)}, pages 220--226, 2015.

\bibitem{Beimel2002BreakingTO}
Amos Beimel, Yuval Ishai, Eyal Kushilevitz, and Jean-Fran\c{c}ois Raymond.
\newblock Breaking the $o(n^{1/(2k-1)})$ barrier for information-theoretic
  private information retrieval.
\newblock {\em The 43rd Annual IEEE Symposium on Foundations of Computer
  Science}, pages 261--270, 2002.

\bibitem{brakerski_fully_2012}
Zvika Brakerski.
\newblock Fully homomorphic encryption without modulus switching from classical
  {GapSVP}.
\newblock In Reihaneh Safavi-Naini and Ran Canetti, editors, {\em Advances in
  Cryptology – {CRYPTO} 2012}, volume 7417 of {\em LNCS}, pages 868--886,
  Berlin, Heidelberg, 2012. Springer.

\bibitem{brakerski_leveled_2012}
Zvika Brakerski, Craig Gentry, and Vinod Vaikuntanathan.
\newblock (leveled) fully homomorphic encryption without bootstrapping.
\newblock In {\em Proceedings of the 3rd Innovations in Theoretical Computer
  Science Conference}, ITCS '12, page 309–325, New York, NY, USA, 2012.
  Association for Computing Machinery.

\bibitem{canettiDoublyEfficientPrivate2017a}
Ran Canetti, Justin Holmgren, and Silas Richelson.
\newblock Towards {{Doubly Efficient Private Information Retrieval}}.
\newblock In Yael Kalai and Leonid Reyzin, editors, {\em Theory of
  {{Cryptography}}}, pages 694--726. Springer, 2017.

\bibitem{cao_and_li}
Ling Cao and Yue Li.
\newblock {IPFS} keyword search based on double-layer index.
\newblock In Zhiyuan Zhu and Fengxin Cen, editors, {\em International
  Conference on Electronic Information Engineering and Computer Communication
  (EIECC 2021)}, volume 12172, page 1217209. International Society for Optics
  and Photonics, SPIE, 2022.

\bibitem{chenOnionRingORAM2019}
Hao Chen, Ilaria Chillotti, and Ling Ren.
\newblock Onion ring {ORAM}: Efficient constant bandwidth oblivious {RAM} from
  (leveled) {TFHE}.
\newblock In {\em Proceedings of the 2019 ACM SIGSAC Conference on Computer and
  Communications Security}, CCS '19, page 345–360, New York, NY, USA, 2019.
  Association for Computing Machinery.

\bibitem{chor1995private}
Benny Chor, Oded Goldreich, Eyal Kushilevitz, and Madhu Sudan.
\newblock {Private Information Retrieval}.
\newblock In {\em Proceedings of IEEE 36th Annual Foundations of Computer
  Science}, pages 41--50. IEEE, 1995.

\bibitem{gibbs2020pirsublinear}
Henry {Corrigan-Gibbs} and Dmitry Kogan.
\newblock Private {{Information Retrieval}} with {{Sublinear Online Time}}.
\newblock In Anne Canteaut and Yuval Ishai, editors, {\em Advances in
  {{Cryptology}} -- {{EUROCRYPT}} 2020}, pages 44--75. Springer, 2020.

\bibitem{bitswap-privacy}
Erik Daniel, Marcel Ebert, and Florian Tschorsch.
\newblock Improving {Bitswap} privacy with forwarding and source obfuscation.
\newblock In {\em 2023 IEEE 48th Conference on Local Computer Networks (LCN)},
  pages 1--4. IEEE, 2023.

\bibitem{bitswap-privacy2}
Erik Daniel and Florian Tschorsch.
\newblock Privacy-enhanced content discovery for {Bitswap}.
\newblock In {\em 2023 IFIP Networking Conference (IFIP Networking)}, pages
  1--9. IEEE, 2023.

\bibitem{frodopir}
Alex Davidson, Gon\c{c}alo Pestana, and Sof\'{i}a Celi.
\newblock {FrodoPIR}: Simple, scalable, single-server private information
  retrieval.
\newblock {\em Proceedings on Privacy Enhancing Technologies},
  2023(1):365--–383, 2023.

\bibitem{de2024whispir}
Leo de~Castro, Kevin Lewi, and Edward Suh.
\newblock {WhisPIR}: Stateless private information retrieval with low
  communication.
\newblock Cryptology ePrint Archive, 2024/266, 2024.
\newblock \url{https://eprint.iacr.org/2024/266}.

\bibitem{bitswap}
Alfonso De~la Rocha, David Dias, and Yiannis Psaras.
\newblock Accelerating content routing with {Bitswap}: A multi-path file
  transfer protocol in {IPFS} and {Filecoin}.
\newblock Technical report, Protocol Labs, 2021.

\bibitem{raidpir}
Daniel Demmler, Amir Herzberg, and Thomas Schneider.
\newblock Raid-pir: Practical multi-server pir.
\newblock In {\em Proceedings of the 6th Edition of the ACM Workshop on Cloud
  Computing Security}, CCSW '14, pages 45--56, New York, NY, USA, 2014.
  Association for Computing Machinery.

\bibitem{fan2012somewhat}
Junfeng Fan and Frederik Vercauteren.
\newblock Somewhat practical fully homomorphic encryption.
\newblock Cryptology ePrint Archive, 2012/144, 2012.
\newblock \url{https://eprint.iacr.org/2012/144}.

\bibitem{usenix2024-crypto-from-research-to-prod}
Konstantin Fischer, Ivana Trummov{\'a}, Phillip Gajland, Yasemin Acar, Sascha
  Fahl, and Angela Sasse.
\newblock The challenges of bringing cryptography from research papers to
  products: Results from an interview study with experts.
\newblock In {\em 33rd USENIX Security Symposium (USENIX Security 24)},
  Philadelphia, PA, 2024. USENIX Association.

\bibitem{space-daemon}
Fleek.
\newblock Overview.
\newblock \url{https://docs.fleek.co/space-daemon/overview/}.
\newblock Accessed on 2023-03-10.

\bibitem{fleek}
Fleek.
\newblock Site deployment.
\newblock \url{https://docs.fleek.co/hosting/site-deployment/}.
\newblock Accessed on 2023-03-10.

\bibitem{gentryFullyHomomorphicEncryption2012}
Craig Gentry, Shai Halevi, and Nigel~P. Smart.
\newblock Fully homomorphic encryption with polylog overhead.
\newblock In David Pointcheval and Thomas Johansson, editors, {\em Advances in
  Cryptology -- {EUROCRYPT} 2012}, volume 7237 of {\em LNCS}, pages 465--482,
  Berlin, Heidelberg, 2012. Springer-Verlag.

\bibitem{paillier-parameters}
Damien Giry.
\newblock Cryptographic key length recommendation.
\newblock \url{https://www.keylength.com/en/4/}.
\newblock Accessed on 2024-02-08.

\bibitem{protobufs}
{Google}.
\newblock What problems do protocol buffers solve?
\newblock \url{https://protobuf.dev/overview/#solve}.

\bibitem{flame-graphs}
Brendan Gregg.
\newblock Flame graphs.
\newblock \url{https://www.brendangregg.com/flamegraphs.html}.

\bibitem{ipfs-chunking-masters-thesis}
Marcel Gregoriadis.
\newblock Analysis and comparison of deduplication strategies in {IPFS}.
\newblock Master's thesis, Humboldt-Universit{\"a}t zu Berlin, 2023.

\bibitem{simplepir}
Alexandra Henzinger, Matthew~M. Hong, Henry Corrigan-Gibbs, Sarah Meiklejohn,
  and Vinod Vaikuntanathan.
\newblock One server for the price of two: Simple and fast single-server
  private information retrieval.
\newblock In {\em 32nd USENIX Security Symposium (USENIX Security 23)}, pages
  3889--3905, Anaheim, CA, 2023. USENIX Association.

\bibitem{aes-aez}
Viet~Tung Hoang, Ted Krovetz, and Phillip Rogaway.
\newblock Robust authenticated-encryption {AEZ} and the problem that it solves.
\newblock In E.~Oswald and M.~Fischlin, editors, {\em Advances in Cryptology --
  {EUROCRYPT} 2015}, volume 9056 of {\em LNCS}, pages 15--44, Berlin,
  Heidelberg, 2015. Springer.

\bibitem{rfc8484}
P.~Hoffman and P.~McManus.
\newblock {DNS Queries over HTTPS (DoH)}.
\newblock RFC 8484 (Proposed Standard), October 2018.

\bibitem{edns}
Z.~Hu, L.~Zhu, J.~Heidemann, A.~Mankin, D.~Wessels, and P.~Hoffman.
\newblock Specification for {DNS} over {Transport Layer Security} ({TLS}).
\newblock RFC 7858 (Proposed Standard), May 2016.
\newblock Updated by RFC 8310.

\bibitem{jaeger-architecture}
The {Jaeger Authors}.
\newblock Architecture.
\newblock \url{https://www.jaegertracing.io/docs/latest/architecture/}, 2023.

\bibitem{hkdf}
Hugo Krawczyk.
\newblock Cryptographic extraction and key derivation: The {HKDF} scheme.
\newblock In Tal Rabin, editor, {\em Advances in Cryptology -- CRYPTO 2010},
  volume 6223 of {\em LNCS}, pages 631--648, Berlin, Heidelberg, 2010.
  Springer.

\bibitem{ipfs-docs}
Protocol Labs.
\newblock How {IPFS} works.
\newblock \url{https://docs.ipfs.tech/concepts/how-ipfs-works/}, 2023.
\newblock Accessed on 2024-02-04.

\bibitem{liHintlessSingleServerPrivate2024}
Baiyu Li, Daniele Micciancio, Mariana Raykova, and Mark {Schultz-Wu}.
\newblock Hintless {{Single-Server Private Information Retrieval}}.
\newblock In Leonid Reyzin and Douglas Stebila, editors, {\em Advances in
  {{Cryptology}} -- {{CRYPTO}} 2024}, pages 183--217. Springer, 2024.

\bibitem{linDoublyEfficientPrivate2023}
Wei-Kai Lin, Ethan Mook, and Daniel Wichs.
\newblock Doubly {{Efficient Private Information Retrieval}} and {{Fully
  Homomorphic RAM Computation}} from {{Ring LWE}}.
\newblock In {\em Proceedings of the 55th {{Annual ACM Symposium}} on
  {{Theory}} of {{Computing}}}, pages 595--608, Orlando FL USA, June 2023. ACM.

\bibitem{lu2019end}
Chaoyi Lu, Baojun Liu, Zhou Li, Shuang Hao, Haixin Duan, Mingming Zhang,
  Chunying Leng, Ying Liu, Zaifeng Zhang, and Jianping Wu.
\newblock An end-to-end, large-scale measurement of {DNS}-over-{Encryption}:
  How far have we come?
\newblock In {\em Proceedings of the 2019 Internet Measurement Conference
  (IMC'19)}, pages 22--35, 2019.

\bibitem{lyubashevskyIdealLatticesLearning2010}
Vadim Lyubashevsky, Chris Peikert, and Oded Regev.
\newblock On ideal lattices and learning with errors over rings.
\newblock In Henri Gilbert, editor, {\em Advances in Cryptology -- {EUROCRYPT}
  2010}, volume 6110 of {\em LNCS}, pages 1--23, Berlin, Heidelberg, 2010.
  Springer.

\bibitem{cpir}
Rasoul~Akhavan Mahdavi and Florian Kerschbaum.
\newblock Constant-weight {PIR}: Single-round keyword {PIR} via constant-weight
  equality operators.
\newblock In {\em 31st USENIX Security Symposium (USENIX Security 22)}, pages
  1723--1740, Boston, MA, 2022. USENIX Association.

\bibitem{kademlia}
Petar Maymounkov and David Mazi\`{e}res.
\newblock {Kademlia}: A peer-to-peer information system based on the {XOR}
  metric.
\newblock In {\em Revised Papers from the First International Workshop on
  Peer-to-Peer Systems}, IPTPS '01, pages 53--65, Berlin, Heidelberg, 2002.
  Springer-Verlag.

\bibitem{dhtpir}
Miti Mazmudar, Stan Gurtler, and Ian Goldberg.
\newblock Do you feel a chill? using {PIR} against chilling effects for
  censorship-resistant publishing.
\newblock In {\em Proceedings of the 20th Workshop on Workshop on Privacy in
  the Electronic Society}, WPES '21, page 53–57, New York, NY, USA, 2021.
  Association for Computing Machinery.

\bibitem{spiral}
Samir~Jordan Menon and David~J. Wu.
\newblock {SPIRAL}: Fast, high-rate single-server {PIR} via {FHE} composition.
\newblock In {\em 2022 IEEE Symposium on Security and Privacy (SP)}, pages
  930--947, 2022.

\bibitem{menonYPIRHighThroughputSingleServer2024a}
Samir~Jordan Menon and David~J. Wu.
\newblock {YPIR}: {High-Throughput} {Single-Server}{PIR} with {Silent
  Preprocessing}.
\newblock In {\em 33rd {USENIX Security Symposium} ({USENIX Security} 24)},
  pages 5985--6002, 2024.

\bibitem{double-hashing-video}
Guillaume Michel.
\newblock Double-hashing as a way to increase reader privacy.
\newblock \url{https://www.youtube.com/watch?v=VBlx-VvIZqU}.
\newblock Accessed on 2024-02-08.

\bibitem{double-hashing}
Guillaume Michel.
\newblock {IPIP-373: Double Hash DHT Spec}.
\newblock \url{https://github.com/ipfs/specs/pull/373}.
\newblock Accessed on 2024-02-04.

\bibitem{onionpir}
Muhammad~Haris Mughees, Hao Chen, and Ling Ren.
\newblock {OnionPIR}: Response efficient single-server {PIR}.
\newblock In {\em Proceedings of the 2021 ACM SIGSAC Conference on Computer and
  Communications Security}, CCS '21, page 2292–2306, New York, NY, USA, 2021.
  Association for Computing Machinery.

\bibitem{okadaPracticalDoublyEfficientPrivate}
Hiroki Okada, Rachel Player, Simon Pohmann, and Christian Weinert.
\newblock Towards {Practical Doubly-Efficient Private Information Retrieval}.
\newblock In {\em Financial Cryptography and Data Security 2024 (FC '24)},
  2024.

\bibitem{opentelemetry}
The {OpenTelemetry Authors}.
\newblock What is opentelemetry?
\newblock \url{https://opentelemetry.io/docs/concepts/what-is-opentelemetry/}.

\bibitem{trace-as-dag-of-spans}
The {OpenTelemetry Authors}.
\newblock Traces.
\newblock
  \url{https://opentelemetry.io/docs/reference/specification/overview/#traces},
  2023.

\bibitem{zipkin}
OpenZipkin.
\newblock Zipkin.
\newblock \url{https://github.com/openzipkin/zipkin/}.

\bibitem{paillier}
Pascal Paillier.
\newblock Public-key cryptosystems based on composite degree residuosity
  classes.
\newblock In Jacques Stern, editor, {\em Advances in Cryptology --
  {EUROCRYPT}'99}, volume 1592 of {\em LNCS}, pages 223--238, Berlin,
  Heidelberg, 1999. Springer.

\bibitem{distributed-tracing-book}
Austin Parker, Daniel Spoonhower, Jonathan Mace, and Rebecca Isaacs.
\newblock {\em Distributed Tracing in Practice}.
\newblock O'Reilly Media, Inc., 2020.

\bibitem{ipfs-chunking}
{Protocol Labs}.
\newblock Chunking.
\newblock \url{https://docs.ipfs.tech/concepts/file-systems/#chunking}.
\newblock Accessed on 2024-02-04.

\bibitem{ipfs-ecosystem}
{Protocol Labs}.
\newblock Ecosystem directory.
\newblock \url{https://ecosystem.ipfs.tech/}.
\newblock Accessed on 2023-03-09.

\bibitem{ipfs-privacy-encryption}
{Protocol Labs}.
\newblock {Privacy and Encryption}.
\newblock \url{https://docs.ipfs.tech/concepts/privacy-and-encryption/}.
\newblock Accessed on 2024-02-04.

\bibitem{teoth}
Bernd Pr{\"u}nster, Alexander Marsalek, and Thomas Zefferer.
\newblock Total eclipse of the heart {\textendash} disrupting the
  {InterPlanetary} file system.
\newblock In {\em 31st USENIX Security Symposium (USENIX Security 22)}, pages
  3735--3752, Boston, MA, August 2022. USENIX Association.

\bibitem{ccr}
Siddhartha Sen and Michael~J. Freedman.
\newblock Commensal cuckoo: secure group partitioning for large-scale services.
\newblock {\em SIGOPS Oper. Syst. Rev.}, 46(1):33–39, 2012.

\bibitem{chord}
Ion Stoica, Robert Morris, David Karger, M.~Frans Kaashoek, and Hari
  Balakrishnan.
\newblock Chord: A scalable peer-to-peer lookup service for internet
  applications.
\newblock In {\em Proceedings of the 2001 Conference on Applications,
  Technologies, Architectures, and Protocols for Computer Communications},
  {SIGCOMM} '01, pages 149--160, 2001.

\bibitem{dennis-providers-study}
Dennis Trautwein.
\newblock {2022-09-20 Hydras Analysis}.
\newblock
  \url{https://github.com/probe-lab/network-measurements/blob/master/results/rfm21-hydras-performance-contribution.md#provider-distribution}.
\newblock Accessed on 2024-02-04.

\bibitem{nebula}
Dennis Trautwein.
\newblock Nebula.
\newblock \url{https://github.com/dennis-tra/nebula}.
\newblock Accessed on 2024-01-25.

\bibitem{ipfsdesign}
Dennis Trautwein, Aravindh Raman, Gareth Tyson, Ignacio Castro, Will Scott,
  Moritz Schubotz, Bela Gipp, and Yiannis Psaras.
\newblock Design and evaluation of {IPFS}: a storage layer for the
  decentralized web.
\newblock In {\em Proceedings of the ACM SIGCOMM 2022 Conference}, SIGCOMM '22,
  pages 739--752, New York, NY, USA, 2022. Association for Computing Machinery.

\bibitem{rcp}
Maxwell Young, Aniket Kate, Ian Goldberg, and Martin Karsten.
\newblock Practical robust communication in {DHTs} tolerating a byzantine
  adversary.
\newblock In {\em 2010 IEEE 30th International Conference on Distributed
  Computing Systems}, pages 263--272. IEEE, 2010.

\bibitem{zhouPIANOExtremelySimple2023}
Mingxun Zhou, Andrew Park, Wenting Zheng, and Elaine Shi.
\newblock {{PIANO}}: {{Extremely Simple}}, {{Single-Server PIR}} with
  {{Sublinear Server Computation}}.
\newblock In {\em 2024 {{IEEE Symposium}} on {{Security}} and {{Privacy}}
  ({{SP}})}, pages 54--54. IEEE Computer Society, October 2023.

\end{thebibliography}

\appendix
\begin{appendices}

\section{Convergence of Private Routing}
\label{app.proof}
\begin{lemma}[Convergence of Private Routing]
If \ipfskadem{} converges in $\lceil \log n \rceil + c_1$ hops, then \privipfskadem{} converges in $\lceil \log n \rceil + c_2$ hops for some constants $c_1, c_2$.
\end{lemma}

\begin{proof}
Consider only the number of hops it takes to reach the target peer given the \emph{closest} peer in each iteration of the protocol. The argument applies equally to the next $k-1$ closest peers.

If the index of the target bucket is greater than the number of rows in the original routing table, then \privipfskadem{} returns the $k$ closest nodes to the server. This is equivalent to the output of \ipfskadem{}, and as such, does not increase the total number of hops. Therefore, we are only concerned with the case where the target index is less than the total number of rows in the original RT and the corresponding bucket is empty. 

For each hop, if the target peer ID is in a non-empty bucket, then \privipfskadem{} returns the same value as \ipfskadem{}. 
If this holds for each hop, then by assumption, \privipfskadem{} converges in $\lceil \log n \rceil + c_1$ hops. 
Otherwise, there are some hops for which the target peer ID is in an \emph{empty} bucket. 
The bound on \ipfskadem{} is assumed to follow by the fact that if the target ID is in a non-empty bucket, then the lookup procedure will return a peer which is at least half as close (i.e., whose distance is at least one bit shorter) than the current peer. 
If this is always the case, then the procedure terminates in $\log n$ steps. 
Therefore, for the algorithm to terminate in $\lceil \log n \rceil + c_1$, the number of times that the lookup procedure does not return a peer which is at least half as close than the current peer must be bounded by the constant $c_1$. That is, the number of times the target ID is in an empty bucket must be bounded by $c_1$.
Suppose then, that there is some constant, $c_1$, number of hops for which the target peer ID is in an \emph{empty} bucket and \privipfskadem{} returns a different peer than \ipfskadem{}. 

Thus, the problem reduces to the case where the following conditions hold: there are $c_1$ hops wherein the target index corresponds to an empty bucket which is less than the total number of rows in the original routing table. In this case, \privipfskadem{} returns first the closer nodes to itself (between the target index and the longest CPL in the RT) and then resorts to returning nodes farther from itself (between the target index and 0).

If \privipfskadem{} returns a node closer to itself, then this must necessarily be no farther from the target node than the server node itself. By assumption, this occurs at most $c_1$ times, and so \privipfskadem{} converges in at most $\lceil \log n \rceil + c_1^2$ hops. Otherwise, \privipfskadem{} returns a peer which is farther from itself; however, the returned peer must have the same CPL with the target node as the peer which would have been returned by \ipfskadem{}. Therefore, this case additionally does not increase the number of hops by more than a constant amount.
\end{proof}

In practice we fix a random seed in our implementation, making the selection of random peers within a range deterministic. This has no effect on the efficiency of the algorithm since all peers within that range are equally likely to be a given distance from the target. Nor does it affect privacy of the client since it does not rely on the target peer ID.

\section{Correctness and Security of PaillierPIR}
\label{sec:prove-paillier}

\subsection{Preliminaries}
\label{app:pir-definition}
We use the following definitions of PIR protocols from \cite{simplepir}, with the Setup algorithm omitted (as it is not relevant for our protocols).

\begin{definition}[PIR protocol]\label{def:pir}
A \emph{PIR protocol} is a triplet of algorithms $(\query, \response,\allowbreak \extract)$ over record space $\mathcal{D}$ and database size $n$, satisfying the following:
\begin{itemize}
    \item $\query(i) \rightarrow (\st,\qu)$: takes as input an index $i \in [n]$ and outputs a client state $\st$ and a query $\qu$.
    \item $\response(\db,\qu) \rightarrow \ans$: given the database $\db$ and query $\qu$, outputs an answer $\ans$.
    \item $\extract(\st,\ans) \rightarrow d$: given client state $\st$ and an answer $\ans$, outputs a record $d \in \mathcal{D}$.
\end{itemize}
We say that a PIR protocol is \emph{correct} if, for any database $\db=\{d_i\}\in\mathcal{D}^{n}$ and any $i\in[n]$, 
\begin{align*}
    \Pr\left[d_i = d_i' \middle|
        \begin{array}{c}
            (\st, \qu) \leftarrow \query(i), \\
            \ans \leftarrow \response(\db, \qu), \\
            d_i' \leftarrow \extract(\st, \ans)
        \end{array}
    \right] \geq 1-\delta,
\end{align*}
for some negligible $\delta$.
Moreover, we say that the PIR protocol is $\epsilon$-\emph{secure}, if for all polynomial-time adversaries $\adv$ and for all $i,j\in[n]$,
\begin{align*}
    &\left|\Pr[\adv(1^{\lambda},\qu)=1 : (\st, \qu)\leftarrow\query(i)]\right. \\
    & \left. - \Pr[\adv(1^{\lambda},\qu)=1 : (\st, \qu)\leftarrow\query(j)] \right| \leq \epsilon(\lambda),
\end{align*}
for negligible $\epsilon$ in a security parameter~$\lambda$.
\end{definition}

\begin{definition}[Semantic security]
A public-key encryption scheme $\pkeScheme = (\KGen, \Enc, \Dec)$ is $\epsilon$-\textsf{IND-CPA} secure, if for any polynomial-time adversary $\advA$, we have that
\[
    \mathsf{Adv}^{\indcpa}_{\pkeScheme}(\advA) := \left| \Pr[\Gm^{\indcpa}_{\pkeScheme}(\advA) = 1] - 1/2 \right| \leq \epsilon,
\]
where $\Gm^{\indcpa}_{\pkeScheme}$ is defined as in \Cref{fig:security-games}.
\end{definition}

We rely on the semantic security of the Paillier cryptosystem, due to \cite{paillier}.

We note that we can rewrite the definition of PIR security in a game-based notion as follows.

\begin{definition}[PIR security]\label{def:pir-security}
A PIR scheme $\pirprot = (\query,\allowbreak \response, \extract)$ is $\epsilon$-secure, if for any polynomial-time adversary $\advA$, we have that
\[
    \mathsf{Adv}^{\pirsec}_{\pirprot}(\advA) := \left| \Pr[\Gm^{\pirsec}_{\pirprot}(\advA) = 1] - 1/2 \right| \leq \epsilon,
\]
where $\Gm^{\pirsec}_{\pirprot}$ is defined as in \Cref{fig:security-games}.
\end{definition}

\begin{figure}[t]
\begin{subfigure}[b]{0.4\columnwidth}
\noindent
    \begin{cralgorithm}{$\Gm^{\indcpa}_{\pkeScheme}(\advA)$}
        \item $(pk,sk) \getsr \KGen()$
        \item $m_0,m_1 \getsr \advA(pk)$
        \item $b \getsr \{0,1\}$
        \item $c \getsr \Enc(pk,m_b)$
        \item $b' \gets \advA(pk,c)$
        \item return $\llbracket b = b'\rrbracket$
    \end{cralgorithm}
\end{subfigure}%
\begin{subfigure}[b]{0.5\columnwidth}
\noindent
    \begin{cralgorithm}{$\Gm^{\pirsec}_{\pirprot}(\advA)$}
        \item $i_0,i_1 \getsr [n]$
        \item $b \getsr \{0,1\}$
        \item $\st,\qu \getsr \query(i_b)$
        \item $b' \gets \advA(\qu)$
        \item return $\llbracket b = b'\rrbracket$
    \end{cralgorithm}
\end{subfigure}%
\caption{Security games for PIR and semantic security}
\label{fig:security-games}
\end{figure}

\subsection{Proof of Correctness}

\begin{theorem}
    \label{thm:basicpaillier-correct}
    Let $M$ denote the composite modulus of Paillier.
    For any $n,\ell\in\NN$, any database $\db\in\ZZ_{M}^{n\times\ell}$ any $i\in\{1,2,\cdots,n\}$, if
    \begin{align}
                (\sk, (\ct, \texttt{p})) \leftarrow \basicpaillier.\query(i) \\
                \ans \leftarrow \basicpaillier.\response(\db, (\ct, \texttt{p})) \\
                d' \leftarrow \basicpaillier.\extract(\sk, \ans)
    \end{align}
    then $d' = \db[i]$.
\end{theorem}
\begin{proof}
    Given that $d',\db[i]\in\ZZ_{M}^{\ell}$, we will prove the theorem by proving that  $d'[k] = \db[i][k]$ for every $k\in[\ell]$.
    First, using the additive properties of Paillier, we see that for every $j\in[n]$,
    \begin{align}
         &\ \pailliercrypto.\decrypt(\sk, \ct'[j]) \\ 
         =&\ \pailliercrypto.\decrypt(\sk, \pailliercrypto.\add(\ct[j], \texttt{p}[j])) \\
        =&\ \pailliercrypto.\decrypt(\sk, \ct[j]) + \texttt{p}[j] \\
        =&\ r_j + (b_j - r_j) = b_j
    \end{align}
    where operations are modulo $M$ in the above equations.
    Moreover, for every $k\in[\ell]$,
    \begin{align}
        d'[k]
        &= \pailliercrypto.\decrypt(\sk, \ans[k]) \\ 
        &= \sum_{j=1}^{n} \db[j][k]\cdot\pailliercrypto.\decrypt(\sk, \ct'[j])\\ 
        &= \sum_{j=1}^{n} \db[j][k]\cdot b_j = \db[i][k]  
    \end{align}
    where the last equation is due to the fact that $b_j = \mathbb{I}[j=i]$.
    Combining the result for all values of $k\in[\ell]$, we observe that $d'=\db[i]$. 
\end{proof}

\subsection{Proof of Security}\label{sec:pir-security}

\begin{theorem}\label{thm:basicpaillier-secure}
    Let $\query$ be as defined in $\basicpaillier.\query$ of $\Cref{alg:basic-pir}$. 
    For any adversary $\adv$,
    we construct an algorithm $\mathcal{B}$ such that
    \[
        \mathsf{Adv}^{\pirsec}_{\basicpaillier}(\advA) = \mathsf{Adv}^{\indcpa}_{\pailliercrypto}(\advB).
    \]
\end{theorem}
\begin{proof}
    We reduce the security of the PIR scheme to the \textsf{IND-CPA} security of the Paillier cryptosystem.
    In particular, for any adversary $\advA$ attacking the PIR security of $\basicpaillier{}$, we construct a new adversary $\advB$ attacking the \textsf{IND-CPA} security of Paillier. 
    Let $d \getsr \{0,1\}$ and $b \getsr \{0,1\}$ denote the challenge bits in the $\pirsec$ and $\indcpa$ games respectively.

    The adversary $\advB$ selects two random indicator vectors $\ell_0$, $\ell_1$ and sends these as messages to the $\Enc$ challenge which returns $c \gets \Enc(pk,\ell_b)$, where $b$ is the challenge bit for $\advB$.
    Now, $\advB$ selects a random $r \getsr \ZZ_{M}^{n}$ from the plaintext space and computes $\widehat{r} = \pailliercrypto.\Enc(pk,r)$. 
    $\advB$ then calls $\advA$ on the following input: $\qu \gets (\pailliercrypto.\add(c,\widehat{r}),-r)$.
    Upon receiving $\advA$'s decision bit, say $d'$, $\advB$ outputs $b' = d'$.

    Now the advantage of $\advB$ is:
    \begin{align*}
        \mathsf{Adv}^{\indcpa}_{\pailliercrypto}(\advB) = \lvert \Pr[d'=1 \mid b = 0] - \Pr[d'=1 \mid b=1] \rvert.
    \end{align*}

    The adversary $\advB$ perfectly simulates the $\pirsec$ game to adversary $\advA$ as it provides a query consisting of:
    \begin{itemize}[leftmargin=8pt]
        \item a ciphertext, $\pailliercrypto.\add(c,\widehat{r})$, indistinguishable from a random ciphertext because $r$ is a random plaintext; and 
        \item a random plaintext $r$,
    \end{itemize}
    such that $\pailliercrypto.\Dec(\pailliercrypto.\add(c,\widehat{r})) + r$ is an indicator vector. It follows that 
    \begin{align*}
        \mathsf{Adv}^{\indcpa}_{\pailliercrypto}(\advB) &= \lvert \Pr[d'=1 \mid d = 0] - \Pr[d'=1 \mid d=1] \rvert \\
        &= \mathsf{Adv}^{\pirsec}_{\basicpaillier}(\advA)
    \end{align*}


\end{proof}



    


\section{Security of Private Provider Advertisements}
\label{sec:security-proof-ce}
We formalize the security requirements and proof for the private provider advertisements protocol (\Cref{alg:non-trivial-private-provider-records}).
The results here also apply to the protocol used for \privwanthave{}, since the algorithm is the same.

The standard PIR security of the scheme (i.e., $\pirsec$ in \Cref{fig:security-games}) follows directly from the security of the underlying PIR scheme, since the server only receives a PIR query for a given row in a table.
We additionally require that this scheme satisfies a \emph{symmetric} property, such that the client does not learn more than the information for which they query.
We can capture this idea by an indistinguishability notion, where the client cannot distinguish between a response generated from the true database and one generated from a random database (aside from the record for which they queried).
This notion is given in the security game in \Cref{fig:spir-sec}, and is equivalent to notions of security for oblivious transfer protocols.

\begin{figure}[ht]
\centering
\begin{minipage}[b]{0.7\columnwidth}
    \begin{cralgorithm}{$\Gm^{\spirsec}_{\pirprot, \mathcal{D}}(\advA)$}
        \item $i \getsr [n], b \getsr \{0,1\}$
        \item $v \gets \mathcal{D}[i], \mathcal{D}_1 \gets \mathcal{D}$
        \item $\mathcal{D}_0 \getsr \{0,1\}^{|\mathcal{D}|}$ \Comment{random database}
        \item $\mathcal{D}_0[i] \gets v$
        \item $\st,\qu \gets \query(i)$
        \item $\ans \gets \response(\mathcal{D}_b, \qu)$
        \item $b' \gets \advA(\ans)$
        \item return $\llbracket b = b'\rrbracket$
    \end{cralgorithm}
\end{minipage}
\caption{Security game for symmetric PIR security}
\label{fig:spir-sec}
\end{figure}

\begin{theorem}
    Let \textsc{PrivateProviderRecord} (\textsc{PPR}) be defined as in \Cref{alg:non-trivial-private-provider-records}.
    For any adversary $\mathcal{A}$ against $\spirsec$, and any database $\mathcal{D}$, we give algorithms $\mathcal{B}_1$, $\mathcal{B}_2$ such that:
    \[
    \Adv_{\textsc{PPR}}^{\spirsec}(\mathcal{A}) \leq \Adv_{\mathsf{KDF}}^{\prf}(\mathcal{B}_1) + \Adv_{\mathsf{SE}}^{\indcpa}(\mathcal{B}_2).
    \]
\end{theorem}
\begin{proof}
We proceed via a series of game hops.
Let $\Gm_0$ denote the original $\spirsec$ game for the \textsc{PPR} protocol for an arbitrary database $\mathcal{D}$, i.e., $\Gm_{\textsc{PPR},\mathcal{D}}^{\spirsec}$.

In $\Gm_1$, we modify $\Gm_0$ to use a random function in place of the $\mbox{KDF}$, unless it is called on the value $v$ (the CID being queried).
This in particular replaces the key $k$ with a random value $k^*$, unless it is derived from the CID that is being queried.
We can bound the difference by the PRF security of $\mbox{KDF}$:
\[
\Pr[\Gm_0] - \Pr[\Gm_1] \leq \Adv_{\mathsf{KDF}}^{\prf}(\mathcal{B}_1).
\]
Next, in $\Gm_2$, we replace the output of $\mbox{SE}.\encrypt$ with a random value of the same length -- again, only on evaluations where the key is already random (and not on evaluations where the key is derived from the queried CID).
Since the input key $k^*$ is random, we can bound this hop by the IND-CPA security of $\mbox{SE}.\encrypt$:
\[
\Pr[\Gm_1] - \Pr[\Gm_2] \leq \Adv_{\mathsf{SE}}^{\indcpa}(\mathcal{B}_2).
\]
Now, we can observe that regardless of the challenge bit $b$, the response~$\ans$ given to the adversary is the same.
Hence, $\mathcal{A}$ has no greater advantage than by guessing the challenge bit $b$ and $\Adv_{\textsc{PPR}}^{\Gm_2}(\mathcal{A}) = 0$.
This concludes the proof.
\end{proof}

\section{Variants of Oblivious Expansion}
\label{sec:oblivious-expand-variant}

The oblivious expansion procedure is described in \Cref{alg:oblivious-expand}.
This procedure makes use of the substitution operation over ciphertexts in RLWE-based schemes~\cite{gentryFullyHomomorphicEncryption2012}.
Recall that plaintexts in RLWE-based cryptosystems are polynomials $p(x)\in R_p$.
The substitution operation, $\texttt{Sub}(k, ct, pk)$, computes the encryption of a plaintext $ct'=\encrypt(pt(X^k))$ given the encryption of plaintext $ct=\encrypt(pt(X))$, for a given $k\in\ZZ_{2N}^*$ using the appropriate cryptographic public key $pk$ (explained below).

To perform $\texttt{Sub}(k, \cdot, pk)$, a substitution key or automorphism key is required.
Angel et al.~generate the required keys for all substitutions of form $k=1+N/2^i$ for $i\in[\log_2 N]$, requiring $\log_2 N$ keys in total.
Note that in the case of a small database such as peer routing, we only have 256 rows so the algorithm can terminate in $8=\log_2 256$ invocations of the outer for loop. Hence, only 8 substitution keys are required for \basicrlwe{} for peer routing. 
We refer the reader to Angel et al.~\cite[Appendix A.2]{sealpir} for the proof of correctness of oblivious expansion, given the substitution operation.

de Castro et al.~\cite{de2024whispir} observed that to reduce the number of required keys, we can compute the substitution of an arbitrary $k'$ by successively applying the substitution $k$, i.e., $p(X) \rightarrow p(X^k) \rightarrow p(X^{k^2}) \rightarrow \cdots$.
Their approach only requires the automorphism key for $k$.
In practice, no generator can generate the set $\{1+N/2^i\}_{i\in[\log_2 N]}\subset \ZZ_{2N}^*$. 
Hence, we need to choose two elements in $\ZZ_{2N}^*$ which can generate the entire set.
de Castro et al. follow this approach and for $N=4096$, they choose 3 and 1173 as the two generators. These two elements generate all the necessary substitutions for oblivious expansion with the least number of total substitutions.

While their approach reduces the number of required keys, it requires more base substitutions, which impacts performance.
To find a middle ground, we can use three keys instead of just two.
Consequently, we can compute all the necessary substitutions with fewer base substitutions.
We conduct an efficient grid search of the space, for $N=4096$, while minimizing the objective function stated by de Castro et al., that is, the total number of base substitutions required by the oblivious expansion algorithm.
We find that $3$, $5$, and $1167$ generate the necessary substitutions that we need with the fewest base substitutions.

\begin{algorithm}[!htbp]
    \caption{Oblivious Expansion with the substitution subprocedure}
    \label{alg:oblivious-expand}
    \begin{algorithmic}[1]
    \Procedure{$\oblivexpand$}{$\pk, \ct$}
        \State $cts \leftarrow [\ct]$
        \For {$i\in[\log_2 N]$}
            \For {$j\in[2^i]$}
                \State $c \leftarrow cts[j]$
                \State $t \leftarrow \texttt{Sub}(\pk, N/2^{i}+1, c)$
                \State $cts[j] \leftarrow c_0 + t$
                \State $cts[j+2^i] \leftarrow X^{-2^i}\cdot(c_1 - t)$
            \EndFor
        \EndFor
        \Return $\{c_i\}_{i\in[N]}$
    \EndProcedure
    \end{algorithmic}    
\end{algorithm}
\iffullversion
\section{Extended PIR Analysis}
\label{app:extended-pir-scheme-analysis}

A body of PIR literature aims to minimize the query processing time to sub-linear in the size of the database~\cite{zhouPIANOExtremelySimple2023,gibbs2020pirsublinear, linDoublyEfficientPrivate2023,canettiDoublyEfficientPrivate2017a, okadaPracticalDoublyEfficientPrivate}.
However, this is achieved by performing an offline phase which introduces an extra round of communication to store a client-specific state at either the client~\cite{zhouPIANOExtremelySimple2023} or the server~\cite{linDoublyEfficientPrivate2023}.
Moreover, the communication and computation costs of client-specific states are justified by the amortization assumption. 
Thus, these schemes do not fit both of our constraints. 
On the other hand, public-key doubly efficient PIR constructions can maintain this client-specific state without an additional round of communication; however, they are not yet practical as demonstrated in the literature~\cite{canettiDoublyEfficientPrivate2017a, okadaPracticalDoublyEfficientPrivate}.

PIR protocols with database-dependent hints and client-specific cryptographic keys can be modified to operate in a single round. Database-dependent hints can be sent along with the PIR response, whereas client-specific cryptographic keys can be sent along with the PIR query.
However, the hints add a large communication overhead that can only be amortized when a large number of queries are made to the same server, which is not applicable for our use-case.
For example, for a 1 GB database, SimplePIR and DoublePIR have hints of 121 MB and 16 MB, respectively.
This constraint excludes protocols such as FrodoPIR~\cite{frodopir}, SimplePIR~\cite{simplepir}, and DoublePIR~\cite{simplepir} from our evaluation. 
\fi
\iffullversion
\section{Implementation Details}
\label{app:implementation-details}

\paragraph{Client-side operations for private routing.}
For private routing, the client uses the target peer's CPL with each server peer when forming a query, instead of the entire target peer ID. The client must also perform cryptographic operations for the PIR scheme, on each such message. While current IPFS implementations expect the client to send the same message to each server in its path to the target, 
our method requires that the client send a different message to each server peer for private versions of these algorithms. Engineering this change in IPFS required appropriately abstracting all client-side logic for sending and receiving messages, and executing it through callbacks. 

\paragraph{Peer routing: server response computation.}
In our private peer routing scenario, the server builds the RT dynamically as discussed in \Cref{subsec:kademlia-and-ipfs-routing-table}, and consequently, it typically has $r\leq256$ buckets. 
Our normalization algorithm returns the $r$th bucket when asked to normalize a bucket for any index $t>r$ (Line~\ref{alg:return-last-bucket}, discussed in \Cref{subsec:kademlia-and-ipfs-routing-table}). 
As an optimization in our implementation, the server can compute a PIR response without instantiating the normalized RT for indices $t>r$. 
Our algorithms simply take the sum of indices $t > r$ in the encrypted indicator vector, and then compute the PIR response over the normalized RT with $r$ buckets. 
The client will obtain the correct output for these indices, without the server learning the plaintext indicator bit values.

\paragraph{Provider routing: bin size.}
We choose to use 4096 bins for both use cases,
given that in \basicrlwe{}, the client can encode a query as large as 4096 within one ciphertext with no extra cost.

\paragraph{Cryptographic keys for symmetric encryption.}
We note that the same key should not be used in each different use case of symmetrically encrypting content. We omit details of an implementation but note that the cases should be disambiguated in some way. For example, if using HKDF Extract/Expand paradigm \cite{hkdf}, a reasonable approach would be to include a label in the \emph{context info} field of the Expand function denoting if the key is for content, multiaddress, or CID encryption. This applies to algorithms in \Cref{sec:priv-bitswap} and \ref{sec.priv-alg-for-providers}.

\paragraph{Evaluation: \privblock{} step.}
As in the provider advertisement use-case, \basicpaillier{} remains infeasible due to its high computation overhead. 
The variants of \basicrlwe{} perform negligibly different for this case, so of the variants, we only evaluate \basicrlwe{}.
This protocol does not scale to databases with many content blocks and the correctness guarantee for \basicrlwe{} only holds when the number of content blocks is less than $\approx 65K$, as RLWEPIR was parameterized for smaller databases for other use-cases.
Thus, neither of our new PIR protocols should be used for \privblock{}.

Since implementations of other evaluated schemes do not support a 256KB payload, we extrapolate communication and computation overheads by multiplying the overheads for the largest supported payload size by the number of plaintexts required to send a 256KB payload.
We run these implementations as-is, and we do not optimize or parallelize them. 
We use the SpiralStreamPack variant of Spiral~\cite{spiral}, as it is optimized for large payloads.
\fi
\section{Impact of Churn}
\label{app:churn}
Trautwein et al.~\cite{ipfsdesign} empirically evaluate node and content churn in IPFS. They find that most peers are online for less than 8 hours; that is, few peers are long-lived. 
Based on the rate at which provider advertisements are updated, Trautwein~\cite{dennis-providers-study} observes that up to half of the advertised content blocks are only advertised for a day, which is the default provider advertisement expiry period.
\iffullversion
Less than half of the content blocks are advertised for two days and so on, leading to a long-tail distribution of the content churn rate for provider advertisements and content blocks. 
\fi
We discuss how our private protocols are impacted by these churn rates.

\paragraph{Private peer routing.}
Peer routing algorithms can be impacted by node churn, but not by content churn.
There are three potential consequences of node churn.

First, as peers join the network, the routing table at each server peer may update, and increase the latency of the private peer routing algorithm. 
Despite the fact that node churn is high, we observe that routing tables are not updated as frequently.
This is due to the fact that Kademlia favours retaining long-lived nodes in its routing table~\cite[Section 2.1]{kademlia}.
Nevertheless, even if peers are added to the routing table, re-normalizing it involves no cryptographic operations and thus has negligible impact on the latency incurred by one server peer. 
So, we expect that the latency of our private peer routing algorithm, under node churn, would be close to that without churn, which is depicted in \Cref{fig:peer-routing-runtime}.

Second, as peers leave the network, some of them may have been in path to the target peer, and consequently, the private peer routing algorithm may not converge.
The original peer routing algorithm was specifically designed to handle high churn.
\iffullversion
In particular, Maymounkov and Mazi\'eres original proof of convergence relies on choosing a sufficiently high $k$ such that the likelihood of $k$ peers all being offline at the same time is negligible~\cite[Section 3]{kademlia}.
Similarly, Trautwein et al.~argue that their choice of $k$ accounts for short-lived peers, such that 
storing $k$ peers in each bucket in the routing table leads to reasonably low number of hops to discover peers~\cite[Section 5.3]{ipfsdesign}.
\fi
Our method does not impose any restriction on the choice of $k$, and so the same assumptions here hold.

Third, the private peer routing algorithm may converge but incur high latency, such that by the time the client finds the target peer, it may be offline.
Since our latency overhead for one hop is comparable to the non-private baseline, and given that we do not significantly increase the number of hops for a lookup over the baseline, our latency for an end-to-end lookup with multiple hops, would also be comparable to a non-private baseline.
Therefore, the likelihood of a node leaving the network by the time a lookup is done, would be similar in both private and non-private cases. 

\paragraph{Private provider routing.}
Our provider routing algorithm can be impacted by node churn, as the peers storing the advertisement leave the network. 
It can also be impacted by content churn, as the advertisements themselves expire in a day by default. 
We argued above that our private peer routing algorithm is resilient to node churn; the same argument applies for finding provider peers privately.
Regarding content churn, a new provider advertisement is processed through four main steps: it is joined with the address book, then the result is symmetrically  encrypted, and then binned through cryptographic hashing, and the bin is then re-encoded for PIR. The first step is a non-cryptographic operation  while the next two steps involve commonly used, efficient cryptographic operations; thus, we expect the first three steps to incur negligible latency. 

    For the last step, appending one provider advertisement to a bin would result in a constant number of new plaintexts, and thus, incur a constant overhead to encode those plaintexts for that bin. 
    However, adding a provider advertisement to the largest bin would require padding and re-encoding other bins. 
    We observe that different bins can be encoded in parallel.
    So, we expect that the latency of private provider routing under content churn can be made as close as possible to \Cref{fig:pir-with-binning-runtimes}, with sufficiently parallelized encoding.

\paragraph{Private content retrieval.}
Our content retrieval algorithm incurs a non-negligibly high latency (order of minutes) over the baseline non-private content retrieval algorithm. 
Nevertheless, we note that this latency is still very small in comparison to the content churn rate. 
Thus, we can safely conclude that there would be sufficient time between successive updates to the content store to re-encode the table for PIR schemes. 
As was the case with the provider advertisement store, this re-encoding step can be parallelized.

\iffullversion
\section{Privacy Engineering Takeaways}
\label{sec:peer2pir-pe-takeaways}
We describe four important takeaways from our process to integrate private routing and content retrieval into IPFS. 
These takeaways also apply to other problems where the goal is to integrate PETs, such as differential privacy or multi-party computation, into an existing non-private codebase for a distributed system.

\paragraph{Start with tests \& observability tools.}
We capture data flows within tests by using observability tools such as distributed tracing. 
Distributed tracing (DT) enables the flow of data through different functions and consequently, identify misbehaving functions~\cite{distributed-tracing-book}.
Popular distributed tracing implementations include OpenTelemetry~\cite{opentelemetry}, Jaeger~\cite{jaeger-architecture} and Zipkin~\cite{zipkin}.
Similar to flame graphs used to profile performance~\cite{flame-graphs}, distributed tracing outputs traces for each top-level function invocation. 
Each trace forms a tree of spans~\cite{trace-as-dag-of-spans}, with a span describing a function call.

Traces generated by tests for a given top-level function will quickly reveal the set of functions that are used to implement that feature. 
In particular, we will be able to track how private information is used by these functions. 
Typically, each function that uses the private information in such a trace will need to be re-implemented privately. 
    
Non-private and private implementations of the same interface may need to coexist in a codebase, potentially for backwards compatibility reasons as we discuss below. 
Thus, after implementing these functions privately, distributed tracing can help ensure that the correct (private or non-private) implementation is being called. 
For instance, when a client peer begins a peer, it first consults its own routing table to find other peers to contact. 
This first lookup should continue with the full target peer ID or target CID. 
However, whenever a server peer receives a query from a client peer, we arrange for our private implementation to be called, and we confirm that this is the case through distributed tracing.

Moreover, if any of the newly implemented functions are also present in traces of tests for other features, then these features could also be impacted by integrating the private implementation. 
So, performing end-to-end tests and analyzing their traces allows us to identify cross-feature interactions that would be relevant when re-implementing the above set of functions. %

\paragraph{Preprocessing data structures.} 
Data structures in the non-private codebase often differ from the model used in PETs. 
For example, in our case, we needed to gather data from multiple tables into a single table for PIR.
Specifically, we join the routing table with the address book for private peer routing (\Cref{sec:private-routing-pir-integration}), and we join the provider store with the address book for private provider routing (\Cref{sec.priv-alg-for-providers}). 
We note that this is distinct from any regular preprocessing operations that the PET itself requires, due to changes in the underlying data structures; for instance, PIR schemes require re-encoding a row that is newly added to a table, before PIR queries can be answered. %


\paragraph{Dealing with randomization when sending queries or receiving responses.}
    In non-private implementations of a given top-level feature, a peer may send the same query (or response) to different peers. 
    However, PETS inherently involve randomization; for instance, through encryption or differential privacy.
    For example, in the case of differential privacy, responses are randomized every time.
    Consequently, in the private implementation of that feature, the peer might send different, randomized messages to different peers. 
    As such, this discrepancy may arise in both centralized (client-server) or distributed (peer-to-peer) architectures.
    It may be challenging to send different queries (or responses) to different peers, when the non-private implementation does not expect such behaviour; we provide an example below. 
    
    Current IPFS implementations expect the client to send the same message to each server in its path to the target, for peer and provider routing. 
    However, the private versions of these functionalities require that the client send a different message to each server peer. 
    In particular, these messages are homomorphically encrypted ciphertexts, corresponding to the PIR query. 
    Engineering this change in IPFS required appropriately abstracting all client-side logic for sending and receiving messages, and executing it through callbacks. %
    The callbacks would derive a ciphertext from the non-private message (target peer ID or target CID) for each server peer; they would be executed just before and instead of transmitting the non-private message. 
    In general, addressing such problems may require significant redesign to avoid code smell.

\paragraph{Backwards compatibility with non-private versions.}
We achieved backward compatibility with non-private versions through careful integration with interface definition languages such as Protobufs~\cite{protobufs}. Specifically, we did not alter the non-private versions of peer or provider routing, instead, we used a different Protobuf message to signal that the client uses the private version. This allowed us to test both private and non-private versions simultaneously for correctness. 

    In general, backwards compatibility enables gradual deployment of private implementations as peers gradually update their code.
    While a peer using an older version and one using a newer version would only be able to communicate using a non-private implementation, pairs of peers who use the newer versions should communicate using the private implementation by default. 
    Consequently, they will be able to obtain the full privacy guarantees of our work for those messages. 

    We highlight that PETS for distributed hash tables should not only be backwards compatible at a protocol level, but also from the perspective of being able to handle data that is already stored in the DHT. 
    In particular, backwards compatible PETS in a distributed setting should not require existing content to be republished. 
    As an example of a violation of this property, we observe that Michel's double hashing approach~\cite{double-hashing}, which we discussed in \Cref{sec.relatedwork}, requires republishing existing content on IPFS.

\paragraph{Interoperability across PET implementations.} 
During our prototyping stage, we want to be able to integrate different implementations of a PET into our system.
    We design our private interfaces to be sufficiently generic so that we can instantiate different PIR  schemes, such as \basicpaillier{} and \basicrlwe{}, under them.
As such, we did not integrate external PET implementations, such as \textsc{Spiral}, into Bitswap, due to time constraints. 
We find that some PET implementations, out of the ones we analyzed in \Cref{subsec:pir-existing-work}, may be more robust or otherwise more fit for integration into Bitswap than others. 

PETs may be implemented in a different programming language than the one used in the codebase.
PETs designed in academia may lack robustness in that they may be parameterized for narrow use-cases (for example, for specific numbers of rows for PIR) and may be difficult to build due to outdated dependencies.
We concur with Fischer et al.'s user study~\cite{usenix2024-crypto-from-research-to-prod} in that standardization may initiate a drive to build robust PETs in academia, which can then be easily evaluated for integration in the industry. 
\fi

\end{appendices}


\end{document}